\documentclass[11pt]{article}

\usepackage[letterpaper, margin=1in]{geometry}

\usepackage{amsmath,amsfonts,amsthm,amssymb,color}
\usepackage{fancybox}
\usepackage{framed}
\usepackage{algpseudocode}
\usepackage{comment}
\usepackage{mathrsfs}
\usepackage{hyperref}
\usepackage{cleveref}
\usepackage{subcaption}

\usepackage{enumitem}
\usepackage{color,soul}
\usepackage{multicol}

\usepackage{footnote}
 \usepackage[ruled,vlined,linesnumbered]{algorithm2e}

\newtheorem{theorem}{Theorem}[section]
\newtheorem{corollary}[theorem]{Corollary}
\newtheorem{lemma}[theorem]{Lemma}

\theoremstyle{definition}
\newtheorem{Definition}[theorem]{Definition}

\newtheorem{Remark}[theorem]{Remark}

\begin{document}

\newcommand\relatedversion{}

\def\prob#1#2{\mbox{Pr}_{#1}\left[ #2 \right]}
\def\pvec#1#2{\vec{\mbox{P}}^{#1}\left[ #2 \right]}
\def\expec#1#2{{\mathbb{E}}_{#1}\left[ #2 \right]}
\def\var#1{\mbox{\bf Var}\left[ #1 \right]}

\def\defeq{\stackrel{\mathrm{def}}{=}}
\def\setof#1{\left\{#1  \right\}}
\def\sizeof#1{\left|#1  \right|}

\def\floor#1{\left\lfloor #1 \right\rfloor}
\def\ceil#1{\left\lceil #1 \right\rceil}

\def\dim#1{\mathrm{dim} (#1)}
\def\sgn#1{\mathrm{sgn} (#1)}

\def\union{\cup}
\def\intersect{\cap}
\def\Union{\bigcup}
\def\Intersect{\bigcap}

\def\abs#1{\left|#1  \right|}

\def\norm#1{\left\| #1 \right\|}
\def\smallnorm#1{\| #1 \|}

\def\insertt{\mathsf{insert}}
\def\update{\mathsf{update}}
\def\delete{\mathsf{delete}}
\def\query{\mathsf{query}}

\newcommand{\wrap}[1]{\left(#1\right)}
\newcommand{\brac}[1]{\left[#1\right]}
\newcommand{\nm}[1]{\left\lVert #1\right\rVert}
\newcommand{\set}[1]{\left\{#1\right\}}
\newcommand{\ket}[1]{\left\lvert #1 \right\rangle}
\newcommand{\bra}[1]{\left\langle #1\right\rvert}
\newcommand{\inner}[1]{\left\langle #1\right\rangle}

\newcommand{\heart}{\heartsuit}
\newcommand{\spade}{\spadesuit}
\newcommand{\h}{\hslash}

\newcommand{\Int}{\text{Int}}
\newcommand{\Ext}{\text{Ext}}
\newcommand{\Bd}{\text{Bd}}
\newcommand{\cut}{\setminus}
\renewcommand{\subset}{\subseteq}
\renewcommand{\supset}{\supseteq}

\renewcommand{\Pr}{\text{Pr}}

\newcommand{\class}[1]{\ensuremath{\mathsf{#1}}}
\renewcommand{\P}{\class{P}}
\newcommand{\BPP}{\class{BPP}}
\newcommand{\NP}{\class{NP}}
\newcommand{\coNP}{\class{coNP}}
\newcommand{\AM}{\class{AM}}
\newcommand{\coAM}{\class{coAM}}
\newcommand{\IP}{\class{IP}}

\newcommand{\cc}{\mathrm{CC}}

\renewcommand{\bar}[1]{\overline{#1}}

\newcommand{\Mat}{\text{Mat}}
\newcommand{\mat}[2]{\left[\begin{array}{#1}#2\end{array}\right]}
\newcommand{\infint}{\int_{-\infty}^\infty}

\newcommand{\hlc}[2]{{\sethlcolor{#1}\hl{#2}}}

\newcommand{\eqdef}{\mathrel{\mathop=}:}

\newcommand{\calP}{\mathcal{P}}
\newcommand{\calQ}{\mathcal{Q}}
\newcommand{\calR}{\mathcal{R}}

\def\endpoints{\mathsf{End}}
\def\vol{\mathbf{vol}}
\def\calP{\mathcal{P}}
\def\updatesingle{\mathtt{Update}}
\def\updateseq{\mathtt{UpdateSeq}}
\def\newupdateseq{\mathtt{NewUpdateSeq}}
\def\sparsifier{\mathsf{Sparsifier}}
\def\contract{\mathsf{Contract}}

\title{\Large Fully Dynamic Min-Cut of Superconstant Size in Subpolynomial Time\relatedversion}
\author{Wenyu Jin\thanks{University of Illinois at Chicago. Email: \href{mailto:wjin9@uic.edu}{wjin9@uic.edu}.}
\and Xiaorui Sun\thanks{University of Illinois at Chicago. This project is supported by the National Science Foundation (NSF) grant 2240024.  Email: \href{mailto:xiaorui@uic.edu}{xiaorui@uic.edu}. }
\and Mikkel Thorup\thanks{University of Copenhagen. This project is supported by the VILLUM Foundation grant 16582. Email: \href{mailto:mikkel2thorup@gmail.com}{mikkel2thorup@gmail.com}. }
}

\date{}

\clearpage
\pagenumbering{arabic}
\setcounter{page}{1}

\maketitle

\begin{abstract} \small\baselineskip=9pt 
We present a deterministic fully dynamic algorithm with subpolynomial worst-case time per graph update such that after processing each update of the graph, the algorithm outputs a minimum cut of the graph if the graph has a cut of size at most $c$ for some $c = (\log n)^{o(1)}$.
Previously, the best update time was  $\widetilde O(\sqrt{n})$ for any $c > 2$ and $c = O(\log n)$ [Thorup, Combinatorica'07].

\end{abstract}

\section{Introduction}

In the study of dynamic graphs, we consider the scenario that a graph over a fixed vertex set undergoes edge insertions and deletions. 
For a property of the graph, 
the goal of a dynamic graph algorithm is to build a data structure that can process graph updates and queries of the graph property efficiently, assuming the updates and queries are presented online, meaning that the algorithm needs to process each update or query without knowing anything about the future.


Update and query time can be  categorized into two types:
\emph{worst-case}, i.e.,  the upper bound on the running time of any  update or query operation, 
and \emph{amortized}, i.e., the running time amortized over a sequence of operations. In this paper we focus on worst-case time, so all operation times are assumed to be worst-case unless explicitly specified.

In this paper, we study the fully dynamic minimum $c$-cut problem for a given integer $c$, in which the dynamic algorithm needs to output a minimum edge cut with size at most $c$ of the graph if such a cut exists after receiving each update of the graph. 
The dynamic minimum $c$-cut problem is a generalization of the dynamic $c$-edge connectivity problem.
The dynamic $c$-edge connectivity problem is to determine if the edge connectivity of the graph is at least $c$, where the edge connectivity of a graph is the size of the minimum cut of the graph. 

The dynamic minimum $c$-cut problem has been extensively studied since 1980s, mostly in the context of dynamic edge connectivity problem~\cite{frederickson1985data, galil1991fully, galil1991fully_3, westbrook1992maintaining,
henzinger1995randomized,
eppstein1997sparsification, frederickson1997ambivalent, henzinger1997fully,  henzinger1997sampling, henzinger1999randomized, thorup2000near, holm2001poly, kapron2013dynamic, wulff2013faster, kejlberg2016faster, nanongkai2017dynamica, nanongkai2017dynamic,  wulff2017fully, holm2018dynamic, chuzhoy2019deterministic, aaman2021optimal}. 
The fully dynamic $1$-edge connectivity problem is the classic fully dynamic connectivity problem. For a  graph with $n$ vertices and $m$ edges,
the best known algorithms for the fully dynamic $1$-edge connectivity have 
deterministic polylogarithmic amortized update time by Holm et al.~\cite{holm2001poly},
Monte Carlo randomized polylogarithmic worst-case update time by Kapron et al.~\cite{kapron2013dynamic},
and $n^{o(1)}$ deterministic worst-case update time by Chuzhoy et al.~\cite{chuzhoy2019deterministic}.

The study of fully dynamic 2-edge connectivity dates back to the work by Westbrook and Tarjan~\cite{westbrook1992maintaining} in a context of maintaining 2-edge connected components, where the
$c$-edge-connected components of a graph are the maximal induced subgraphs with edge connectivity at least $c$. 
Galil and Italiano~\cite{galil1991fully} obtained 
an  algorithm
with  $O(m^{2/3})$ update time. 
The update time  was improved to  $ O(\sqrt{m})$ by 
Frederickson~\cite{frederickson1997ambivalent}, and  $O(\sqrt{n})$ by Eppstein et al.~\cite{eppstein1997sparsification}. All these running times are worst-case. 
To date, the best known worst-case update time is still $O(\sqrt{n})$.
For the amortized case, 
polylogarithmic amortized update time algorithm were proposed in~\cite{henzinger1997fully,thorup2000near,holm2001poly,  holm2018dynamic}. 

The dynamic $c$-edge connectivity problem gets much harder for $c>2$. For $c\leq 2$, two vertices are in the same $c$-edge connected component 
iff they are $c$-edge connected (i.e., the two vertices cannot be disconnected by removing fewer than $c$ edges),
but this is not the case for higher $c$. For example, if two vertices are connected by three edge disjoint paths of length two, then the two vertices are $3$-edge connected but the $3$-edge connected components are all singleton vertices.
For the fully dynamic $3$-edge
and $4$-edge connectivity, $O(n^{2/3})$ and $O(n\alpha(n))$ update time algorithms were given by Eppstein et al.~\cite{eppstein1997sparsification}.
Thorup~\cite{thorup2007fully} presented a deterministic fully dynamic minimum $c$-cut algorithm for any polylogarithmic $c$ with $\widetilde O(\sqrt{n})$ time per update. 
Recently, 
Goranci et al. gave two fully dynamic algorithms for arbitrary edge connectivity with $\widetilde O(n)$ worst-case update time and $\widetilde O(m^{1 - 1/ 16})$ amortized update time respectively~\cite{goranci2023fully}.


All of the above results are for exact $c$-edge connectivity. 
For the fully dynamic approximate edge connectivity
problem, 
Thorup
and Karger~\cite{thorup2000dynamic} presented a deterministic $\sqrt{2+o(1)}$-approximate algorithm for arbitrary edge connectivity in amortized polylogarithmic time per update.
Thorup~\cite{thorup2007fully} presented a Monte Carlo randomized $(1 + o(1))$-approximate algorithm in $\widetilde O(\sqrt{n})$ update time, which only works for oblivious adversary.

On the other hand, a closely related problem to the problem studied in this paper is the dynamic $(s, t)$ $c$-edge connectivity problem~\cite{galil1991fully, galil1991fully_3, eppstein1997sparsification, frederickson1997ambivalent, henzinger1997fully, js20}. 
In the dynamic $(s, t)$ $c$-edge connectivity problem, 
graph updates and queries each containing two vertices are presented online so that for each query, the dynamic algorithm needs to determine if the edge connectivity between the two queried vertices is at least $c$, where the edge connectivity between two vertices in a graph is the size of the minimum cut separating the two vertices. 
We remark that the dynamic minimum $c$-cut problem (as well as dynamic $c$-edge connectivity problem) and the dynamic $(s, t)$ $c$-edge connectivity problem are related, but do not reduce to each other efficiently. 
To solve the dynamic minimum $c$-cut problem, 
for any vertex $s$, there exists a vertex $t$ such that the minimum cut separating $s$ and $t$ is a global minimum cut. However, since $t$ could be any vertex in the graph, it is not efficient to use fully dynamic $(s, t)$ $c$-edge connectivity algorithm to solve the dynamic minimum $c$-cut problem.

For the fully dynamic $(s, t)$ $c$-edge connectivity problem, 
polylogarithmic, $O(\sqrt{n})$, and $O(n^{2/3})$ worst-case update and query time algorithms for $c=1, 2,$ and $3$ were presented in~\cite{henzinger1999randomized, frederickson1997ambivalent}, and \cite{galil1991fully_3} respectively. 
But it was unknown if there is a sublinear time fully dynamic $(s, t)$ $c$-edge connectivity algorithm for $c> 3$, until very recently a fully dynamic $(s, t)$ $c$-edge connectivity algorithm with $n^{o(1)}$ update and query time for any $c = (\log n)^{o(1)}$ was proposed in~\cite{js20} by Jin and Sun. 

An obvious open question is if the fully dynamic minimum $c$-cut problem can also be solved in subpolynomial update time. 
In this paper, we give an affirmative answer to this question by proving the following result. 

\begin{theorem}\label{thm:main}
For any $c = (\log n)^{o(1)}$, 
there is a deterministic fully dynamic algorithm for an undirected graph with $n^{o(1)}$ running time per update such that after processing each update of the graph, the algorithm outputs a (global) minimum cut of the graph if the graph has a cut of size at most $c$, or outputs an empty set if the (global) minimum cut of the graph is of size greater than $c$. 
\end{theorem}

Consider the simpler decremental case where we start from a graph with $m$ edges that can only be deleted (no edge insertions). Plugging our fully dynamic algorithm into the decremental reduction from
\cite{aaman2021optimal}, we get
\begin{corollary}
    For any $c = (\log n)^{o(1)}$, there exists a Monte Carlo randomized decremental $c$-edge-connectivity data structure which as in Theorem \ref{thm:main} maintains a minimum cut of size at most $c$ in $O(m+n^{1+o(1)})$ total time.
\end{corollary}
Above, the amortized cost per edge deletions
is constant if we start with more
than $n^{1+o(1)}$ edges.

\subsection{Techniques}
While the minimum cut of a graph is a global property, 
we present a localization approach using a vertex partition of the graph so that the process of finding a minimum cut of the graph with cut size at most $c$ is reduced to finding a minimum cut with certain properties among the induced subgraphs on all the vertex sets in the vertex partition.

We first give a static algorithm to  illustrate the localization idea, and then show how to make the static algorithm fully dynamic. 




\subsubsection{A Static \texorpdfstring{$c$}{c}-Edge Connectivity Algorithm via Localization}

We start with an iterative algorithm for a \emph{static} graph based on the terminal edge connectivity sparsifier, which was proposed in \cite{chalermsook2020vertex, js20}.
For simplicity, the algorithm is for the $c$-edge connectivity problem, ignoring the issue of retrieving a corresponding minimum cut in the dynamic minimum $c$-cut problem.

The static algorithm illustrates the iterative computation structure, and we will show that the iterative computation structure can be maintained efficiently also for a dynamically changing graph.

\vspace{.15cm}\noindent \textbf{Terminal Edge Connectivity Sparsifier \ \ }
We first define terminal edge connectivity sparsifier~\cite{chalermsook2020vertex, js20}. 
Roughly speaking, for a given set of vertices called terminals, the terminal edge connectivity sparsifier is a graph that preserves the minimum cut size for any partition of the terminals. 
Formally, for a graph $G$, a terminal vertex set $T$, and an integer $c$, 
the \emph{terminal $c$-edge connectivity sparsifier} with respect to $G$ and $T$ is a graph containing all the terminals such that for any $\emptyset \subsetneq T' \subsetneq T$, the size of the minimum cut separating $T'$ and $T\setminus T'$ in $G$ is the same as the size of the minimum cut separating $T'$ and $T\setminus T'$ in the sparsifier, if the minimum cut size is smaller than or equal to $c$.
Since a terminal sparsifier only preserves the minimum cut size for terminal partitions, 
a sparsifier can be potentially much smaller than the original graph in terms of the number of vertices and edges so that the $c$-edge connectivity between the terminals can be efficiently determined in the sparsifier.

As we are going to apply a vertex partition to the graph, we define the terminal edge connectivity sparsifier with respect to a vertex partition.  
For a graph $G$ and a vertex partition $\calP$ of $G$, 
the terminal $c$-edge connectivity sparsifier for $G$ and $\calP$ takes the boundary vertices of $\calP$ as terminals, where the boundary vertices of each cluster (i.e., vertex set in $\calP$) are those who have a neighbor in a different cluster. 
We denote the terminals for $G$ and $\calP$ as $T_{G, \calP}$.
A terminal $c$-edge connectivity sparsifier for $G$ and $\calP$ is constructed as follows: start with $G$, for each cluster $P\in \calP$, replace the induced subgraph on $P$ by 
a terminal $c$-edge connectivity sparsifier for the induced subgraph with the boundary vertices in the cluster (i.e., $T_{G, \calP} \cap P$) as terminals.
One can verify that the sparsifier constructed in this way is a terminal $c$-edge connectivity sparsifier for $G$ with $T_{G, \calP}$ as terminals.

Now, suppose we can construct the terminal $c$-edge connectivity sparsifier for a given graph and a vertex partition.
We are going to discuss how to use it to compute the $c$-edge connectivity for a graph. 

\vspace{.15cm}\noindent \textbf{Global $c$-Edge Connectivity via Localization\ \ }
Suppose we have a vertex partition $\calP$ for graph $G$, and a terminal sparsifier for $G$ and $\calP$ that preserves the minimum cut size for any partition of the boundary vertices with respect to $\calP$ for some integer $c$. We show how to compute the $c$-edge connectivity of $G$.  

The high level idea of our algorithm is to ``localize" the computing of $c$-edge connectivity based on a given vertex partition $\calP$ of the input graph $G$. 
By localization, we mean that the process of computing the $c$-edge connectivity of $G$ is reduced to finding a minimum $c$-cut with certain properties among the induced subgraphs on all the clusters of $\calP$. Thus, even though the $c$-edge connectivity is a global graph property, with the localization, it is sufficient to examine the cuts within the induced subgraph on each cluster, without knowing other clusters.

To achieve the localization, we make use of the terminal $c$-edge connectivity sparsifier for $G$ and $\calP$. 
We show that if there is a minimum $c$-cut of $G$ that is a terminal cut (i.e., a cut that partitions terminals into two non-empty sets), then the $c$-edge connectivity of $G$ is the same as the $c$-edge connectivity of the terminal sparsifier. Otherwise, all the minimum $c$-cuts of $G$ are \emph{local cuts},
where a cut $C$ of $G$ is a local cut if it is a non-terminal cut (i.e., a cut such that all the terminals $T_{G,\calP}$ are in the same side of the cut), and the cut-set of the cut is contained in the induced subgraph of a cluster in $\calP$. 

Thus, for a vertex subset $P\in \calP$, let the \emph{local $c$-edge connectivity} for $G[P]$ be the size of the minimum $c$-local cut in $G[P]$ if any.
The $c$-edge connectivity of $G$ is the minimum of the terminal sparsifier's $c$-edge connectivity 
and the local $c$-edge connectivity for all the clusters in $\calP$.

Furthermore, if the induced subgraph on each cluster has high conductance, then each local cut has a side with subpolynomial volume.  
Thus the local cuts can be efficiently enumerated, and the local $c$-edge connectivity in the induced subgraph can be efficiently computed. 
This occurs when $\mathcal{P}$ is a $1/n^{o(1)}$-expander decomposition, i.e., every cluster in $\mathcal{P}$ is a $1/n^{o(1)}$-expander.



\vspace{.15cm}\noindent \textbf{Iterative Construction\ \ }
Notice that in order to compute the $c$-edge connectivity for the graph by the localization, we still need to compute the $c$-edge connectivity of the sparsifier. 
Hence, we have an iterative construction starting from the original graph $G^{(0)} = G$. 
Starting with $i = 0$, we compute an $1 / n^{o(1)}$-expander decomposition $\calP^{(i)}$ of $G^{(i)}$ with terminals $T_{G^{(i)},\calP^{(i)}}$. Next, we compute the local $c$-edge connectivity for each cluster of $\calP^{(i)}$.  We are now done if $\calP^{(i)}$ has only one cluster and no terminals. Otherwise, we construct a terminal $c$-edge connectivity sparsifier $G^{(i+1)}$ of $G^{(i)}$ and $\calP^{(i)}$,  and then repeat the above process on $G^{(i+1)}$. 
The $c$-edge connectivity of $G$ is then the smallest local $c$-edge connectivity over all the clusters in the decomposition $\calP^{(i)}$ for all $i$.

In this paper, we aim to make use of the above iterative construction of expander decomposition and terminal $c$-edge connectivity sparsifier to study the fully dynamic minimum $c$-cut problem. 
To construct a minimum $c$-cut corresponding to the smallest local $c$-edge connectivity found at some cluster $P$ in $\calP^{(i)}$, we have to trace how
the induced subgraph of $G^{(i)}$ on $P$ was created from the lower $G^{(j)}$  and $\calP^{(j)}$ with $j<i$.
If the graph is dynamic, the trace of local cuts becomes more complex, as it is affected by the updates for graph, expander decomposition, and terminals throughout the iterative construction.

We remark that Goranci et al. proposed and dynamically maintained the expander hierarchy, which iteratively computes the expander decomposition and contracts each cluster into a single vertex~\cite{goranci2021expander}. 
However, the work by Goranci et al. can only approximate edge connectivity, and cannot compute the exact edge connectivity, 
which is the issue we aim to address in this paper.


\subsubsection{Retrieving A Minimum \texorpdfstring{$c$}{c}-Cut From Terminal Sparsifier}


Our fully dynamic minimum $c$-cut algorithm makes use of the terminal edge connectivity data structure presented in~\cite{js20} by Jin and Sun, where the terminal edge connectivity data structure was used to solve the fully dynamic $(s, t)$ edge connectivity problem. 

The data structure is based on a construction of terminal $c$-edge connectivity sparsifier for any $c = (\log n)^{o(1)}$ such that the number of vertices and edges in the sparsifier is linear to the number of terminals, and such a sparsifier can be updated in $n^{o(1)}$ time if the input graph has the conductance at least $1 /n^{o(1)}$.
In order to construct and update the terminal $c$-edge connectivity sparsifier efficiently for an arbitrary graph, the efficient update algorithm for the terminal edge connectivity sparsifier of the graphs that are $1 /n^{o(1)}$-expanders is combined with the update algorithms of expander decomposition~\cite{saranurak2019expander, chuzhoy2019deterministic}.

Since the computation of 
the $c$-edge connectivity between two vertices in the sparsifier may still be inefficient if a $1/n^{o(1)}$-expander decomposition is used for the input graph, the sparsification is applied iteratively until the vertex partition contains only a single cluster. 

In this paper, we show that if there is a minimum $c$-cut of $G$ that is a terminal cut, then any minimum $c$-cut of the terminal $c$-edge connectivity sparsifier corresponds to a minimum $c$-cut of $G$, and thus finding a minimum terminal $c$-cut in the graph is reduced to finding a minimum cut in the terminal $c$-edge connectivity sparsifier (Lemma~\ref{lem:sparsifier_terminal_cut}). 
This observation requires careful analysis of the aforementioned terminal $c$-edge connectivity sparsifier construction because the definition of the terminal $c$-edge connectivity sparsifier only promises to preserve the minimum cut size for the terminal cuts. 
Moreover, since the vertex set of the sparsifier does not contain all vertices of $G$, with a minimum $c$-cut of the sparsifier, it is not straightforward to define the corresponding cut of $G$. 
By investigating the construction of the terminal $c$-edge connectivity sparsifier, we build the correspondence between minimum $c$-cuts of the sparsifier and cuts of the same cut size in $G$, and show that given a minimum $c$-cut of the sparsifier, the corresponding terminal cut of the same cut size in $G$ can be computed efficiently.

We remark that although Saranurak also presented an edge connectivity sparsifier based on expander decomposition~\cite{saranurak2021simple}, it is unknown how to efficiently maintain the sparsifier proposed in~\cite{saranurak2021simple} under graph updates. Thus, for the purpose of a dynamic algorithm, we utilize the data structure proposed in~\cite{js20}.

\subsubsection{Minimum \texorpdfstring{$c$}{c}-Cut Data Structure}\label{sec:intro_min_data_structure}

For a graph $G$ and a vertex partition $\mathcal P$ which is a $1/n^{o(1)}$-expander decomposition of $G$, in the previous section, we discussed the way to find a global minimum cut if it is a terminal cut with respect to partition $\mathcal P$. In the rest of this section, we address the remaining problem of finding a global minimum cut which is a local cut.
 

We discuss the data structure used to maintain a minimum $c$-local cut for a graph $G$ and a vertex partition $\calP$. 
The main difficulty of keeping track of the minimum local cuts is that within a single cluster of a vertex partition, the number of local cuts can be $n^{O(c)}$, and these local cuts change significantly even with a single update of the graph. 
In order to update the data structure efficiently, we need to carefully organize the local cuts.


 In our solution, we characterize some necessary conditions for a local cut that can potentially be a global minimum cut of the graph. 
 This characterization reduces the number of local cuts that need to consider from $n^{O(c)}$ to $n \cdot c^{O(c)}$. 
 Furthermore, we show that these local cuts can be enumerated efficiently. 
 We use a priority queue to keep all these local cuts in the ascending order of the cut sizes.
 Thus, if the global minimum cut of the graph is a local cut, it is the first cut in the priority queue. 

The minimum $c$-cut data structure contains a multi-level terminal $c$-edge connectivity data structure, and for each level, the minimum $c$-cut data structure stores an additional priority queue that maintains all local cuts that can potentially be a global minimum cut according to our characterization.  

\subsubsection{Fully Dynamic Minimum \texorpdfstring{$c$}{c}-Cut Algorithm}\label{sec:intro-update}
We give a brief overview of our fully dynamic algorithm. We focus on maintaining the minimum $c$-local cut data structure with respect to the updates of the graph and the vertex partition. 

The main difficulty is that a lot of cuts may switch between local and non-local cuts when the graph gets updated.
When the graph gets updated, the expander decomposition also gets changed. 
Consequently, since the terminals are defined as the boundary vertices of the vertex partition, 
the terminal set is also changed. 
Some previous terminals can become non-terminals (e.g., the incident edges to other clusters are removed), and some previous non-terminals can become terminals (e.g., an incident edge for a non-terminal to another cluster is inserted). 
Since local cuts require that the terminals within the induced subgraph are on the same side of the cuts, 
the change of terminals may turn some non-local cuts to local cuts, and some local cuts to non-local cuts. 
Thus, there can be a large number of $c$-cuts switching between local and non-local cuts after only one update of the graph.

As our main observation, we show that even though the total number of cuts switching between local and non-local can be large after graph updates, 
the number of local cuts in the priority queue that need to be updated is always small based on the characterization in Section~\ref{sec:intro_min_data_structure}.
Furthermore, we can efficiently find all local cuts in the priority queue that need to be updated.



Our solution considers the following offline update scenario: 
Let $G$ be a graph, and $G'$ be a graph after applying graph updates to $G$. 
Let $\calP$ and $\calP'$ be $1 / n^{o(1)}$-expander decomposition of $G$ and $G'$ respectively, such that each cluster of $\calP'$ is a subset of a cluster of $\calP$, and for each cluster $P'\in \calP'$, the induced subgraph of $G'$ on $P'$ is the same as the induced subgraph of $G$ on $P'$. 

Note that before and after the updates, we only study local cuts of size at most $c$ in $1 / n^{o(1)}$ expanders, so local cuts always have a side of volume at most $c \cdot n^{o(1)}$.  As our main observation, we show that if a cut switches from a local cut to a non-local cut, then for the local cut before the update, the side with volume at most $c \cdot n^{o(1)}$ has at least one of the following two types of vertices:
\begin{itemize}
    \item a vertex that switches between a terminal and non-terminal in the update;
    \item a vertex that is an endpoint of an edge inserted or deleted in the update.
\end{itemize}

Similarly, if a cut switches from a non-local cut to a local cut, then for the local cut after the update, at least one of the two aforementioned types of vertices exists on the side of volume at most $c \cdot n^{o(1)}$.

Therefore, we update the priority queue containing local cuts of $G$ to a priority queue containing local cuts of $G'$ in the following way: 
\begin{enumerate}
\item List all of the aforementioned two types of vertices in $G$. These are the vertices that are ``affected'' by the graph updates.
\item Remove all local cuts of $G$ in the priority queue that have an affected vertex on the side of volume at most $c \cdot n^{o(1)}$.
\item Enumerate local cuts of $G'$ that contain an affected vertex on the side of volume at most $c \cdot n^{o(1)}$ and add the enumerated cuts back to the priority queue.
\end{enumerate}

We show the number of ``affected'' vertices is linear in the number of graph updates up to a subpolynomial factor. Therefore, the total number of cuts that need to be removed from and added to the local cut data structures in each level is at most a subpolynomial factor to the number of updates. This allows us to update the data structure efficiently (Lemma \ref{lem:one-up-new}). 
\section{Preliminaries}\label{sec:preliminary}

\vspace{.15cm}\noindent \textbf{Multigraph \ \ }
For a multigraph $G = (V, E)$, the edge set $E$ is a set of triples $(u, v, \alpha)$ such that $u$ and $v$ are two vertices in $V$, and $\alpha$ is a positive integer which represents the edge multiplicity of the edge with $u$ and $v$ as two endpoints. 
We assume $(u, v, \alpha)$ is equivalent to $(v, u, \alpha)$ so that for every two vertices $u$ and $v$, there is at most one triple with $u$ and $v$ as first two elements. 
For a graph $G$, we use $V(G)$ to denote the vertex set of $G$, and $E(G)$ to denote the edge set of $G$. 

We use the degree reduction technique~\cite{harary6graph} to transform the input simple graph with an arbitrary maximum degree to a multigraph such that every vertex has at most a constant number of distinct neighbors. Furthermore, each edge insertion or deletion for the simple graph corresponds to a constant number of multigraph updates with the following update operations:
\begin{itemize}
\item $\insertt(u, v, \alpha)$: insert edge $(u, v, \alpha)$ to the graph
\item $\delete(u, v)$: delete the edge with $u$ and $v$ as two endpoints from the graph, no matter what the edge multiplicity is 
\item $\insertt(v)$: insert a new vertex $v$  to the graph
\item $\delete(v)$: delete isolated vertex $v$ from the graph.
\end{itemize}
See Appendix~\ref{sec:transform} for the detailed discussion on the multigraph construction and update.
In this paper, unless specified, we assume graphs are undirected multigraphs with a constant number of neighbors.

\vspace{.15cm}\noindent \textbf{Notations \ \ }
We say $E'$ is an edge subset of $G = (V, E)$ if $E'$ is a subset of $E$.
For a subset of edges $E' \subset E$, 
we use $|E'|$ to denote $\sum_{(u, v, \alpha)\in E'}\alpha$, use $\endpoints(E')$ to denote the set of endpoints of all the edges in $E'$, and use $G\cut E'$ to denote the graph $(V, E\cut E')$.

For a graph $G = (V, E)$ and a subset of vertices $V'\subset V$, 
we use $\vol_G(V')$ to denote the volume of $V'$ in $G$, which is defined as 
\[\vol_G(V') \defeq \sum_{u\in V'} \sum_{(u, v, \alpha) \in E} \alpha.\] 
We use $G[V']$ to denote the induced subgraph of $G$ on $V'$, and use $\partial_G(V')$ to denote the set of edges with one endpoint in $V'$ and another endpoint in $V\cut V'$. 

Let $\calP$ be a vertex partition of $V$ for graph $G$.
We call each set in $\calP$ a cluster. 
We say an edge of $G$ is an intercluster edge with respect to $\calP$ if the two endpoints of the edge belong to different clusters, otherwise, the edge is an intracluster edge. 
We use $G[\calP]$ to denote the union of the induced subgraphs on all the clusters in $\calP$ on $G$, i.e., $G[\calP] = (V, E')$ where $E' = \{(u, v, \alpha) \in E: \exists P \in \calP \text{ s.t. } u, v\in P\}$. 
We use $\partial_G(\calP)$ to denote the union of $\partial_G(P)$ for all the $P \in \calP$, i.e., the intercluster edges with respect to $\calP$.
A vertex is a boundary vertex with respect to $\calP$ if there exists an intercluster edge incident to the vertex. 

A \textit{cut} is a bipartition of the graph vertices. The \textit{cut-set} of a cut is the set of edges with two endpoints belonging to different sides of the cut. The \textit{size of a cut} is the sum of edge multiplicities for all the edges in the cut's cut-set. 
For a given positive integer $c$, we say a cut is a \textit{$c$-cut} if the size of the cut is at most $c$.

\vspace{.15cm}\noindent \textbf{Expander Decomposition \ \ }
We also use expander decomposition in this paper. 

\def\truefactor{2^{\delta \log^{1/3} n \log^{2/3} \log n}}

\begin{Definition}[Expander and Expander Decomposition]
For a graph $G = (V, E)$ of $n$ vertices and $m$ edges,  the conductance of $G$ is 
$$\min_{\emptyset\subsetneq S\subsetneq V}\frac{\abs{\partial_G(S)}}{\min\set{\vol_G(S), \vol_G(V\cut S)}}$$
A graph $G$ is a $\phi$-expander if the conductance of $G$ is at least $\phi$.

A $(\phi, \epsilon)$-expander decomposition of $G$ is a vertex partition $\mathcal P$ of $V$ such that for each $P\in\mathcal P$, the induced subgraph $G[P]$ is a $\phi$-expander, and the number of intercluster edges is at most $\epsilon \cdot m$. 
For convenience, a $(\phi, \epsilon)$-expander decomposition is also called a $\phi$-expander decomposition. 
\end{Definition}

We utilize the deterministic expander decomposition algorithm by  
Chuzhoy et al.~\cite{chuzhoy2019deterministic}, and the expander pruning algorithm 
by Saranurak and Wang in \cite{saranurak2019expander} used in this paper.

\begin{theorem}[Corollary 7.1 of ~\cite{chuzhoy2019deterministic}]\label{thm:expander_decomposition}
Given a graph $G = (V, E)$ of $m$ edges and a parameter $0 < \phi < 1$,
there is 
a deterministic algorithm  to compute a $\left(\phi, \phi \cdot  \truefactor\right)$-expander decomposition of $G$ in time $(m/ \phi^2)^{1+o(1)}$ for some constant $\delta > 0$.
\end{theorem}

\begin{theorem}[Theorem 1.3 of \cite{saranurak2019expander}, rephrased]\label{thm:pruning}
Let $G = (V, E)$ be a simple $\phi$-expander with $m$ edges. 
Given access to adjacency lists of $G$ and a set $D$ of $k \leq \phi m / 10 $ edges, 
there is a deterministic algorithm to find a pruned set $P \subseteq V$ in time $O(k \log m/\phi^2)$ such that all of the following conditions hold:
\begin{enumerate}
\item $\vol_G(P) = 8 k  / \phi$.
\item $|E_G(P, V \cut P)| \leq 4k$.
\item $G'[V \cut P]$ is a $\phi / 6$ expander, where $G' = (V, E \cut D)$.
\end{enumerate}
\end{theorem}


\section{Review: Multi-level Terminal Edge Connectivity Data Structure}\label{sec:js20_review}

In this section, we review the definition and main properties of the multi-level terminal edge connectivity data structure defined in \cite{js20}. 
The $c$-edge connectivity between a pair of vertices is defined as follows. 
\begin{Definition}
For a multigraph $G = (V, E)$,  two vertices $s, t\in V$, and a positive integer $c$, 
two vertices $s$ and $t$ are $c$-edge connected if $s$ and $t$ are connected in $G\setminus E'$ for all $E'\subset E$ such that $|E'| < c$. 

The \emph{edge connectivity} of $s$ and $t$ for graph $G$, denoted as $\lambda_{G}(s, t)$, is the non-negative integer such that $s$ and $t$ are $\lambda_{G}(s, t)$ edge connected, but not $(\lambda_{G}(s, t)+1)$-edge connected. 

For a positive integer $c$, the \emph{$c$-edge connectivity} of $s$ and $t$ for graph $G$ is $\min\{c, \lambda_{G}(s, t)\}$, and the \emph{$c$-edge connectivity} for graph $G$ is 
\[\min\left\{c, \min_{u, v\in V: u \neq v} \lambda_{G}(s, t)\right\}.\]
\end{Definition}

The dynamic $(s, t)$ edge connectivity problem was studied in~\cite{js20}. 
In this problem, for a fixed positive integer $c$, 
the algorithm receives updates of the graph, as well as $c$-edge connectivity queries each containing two vertices. 
The goal of the algorithm is to process each update and to answer the $c$-edge connectivity for each pair of queried vertices as efficiently as possible. 
As the main result, 
a fully dynamic algorithm of $(s, t)$ $c$-edge connectivity for any $c = (\log n)^{o(1)}$ with $n^{o(1)}$ worst-case  update and query time was presented in~\cite{js20}.   

We remark that the fully dynamic global minimum cut problem, the problem we want to solve in this paper, cannot be solved in subpolynomial time using the fully dynamic $(s, t)$ edge connectivity algorithm. Even though there always exist two vertices in the graph such that the edge connectivity between the two vertices is the same as the global minimum cut size, 
it is non-trivial to find such a vertex pair after each update of the graph, and
the only way we know how to find such a pair of vertices is to first identify the global minimum cut.

\subsection{Terminal Edge Connectivity Sparsifier}
The high level intuition of \cite{js20} is to construct and update a terminal edge connectivity sparsifier for a subset of selected vertices, called terminals, such that the $c$-edge connectivity is preserved for any pair of terminals in the sparsifier compared with the input graph.

\begin{Definition}
Let $G = (V, E)$ be a multiple graph, $T\subset V$ be a terminal set, and $c$ be a positive integer. 
A multigraph $H = (V_H, E_H)$ is a \emph{terminal $c$-edge connectivity sparsifier} of $G = (V, E)$ with respect to $T$ if $T$ is a subset of both $V$ and $V_H$ such that the $c$-edge connectivity between any two vertices $u, v\in T$ for $H$ is the same as the $c$-edge connectivity between $u$ and $v$ for $G$. 

    
\end{Definition}


We use the cut containment set, which was proposed in~\cite{kr13, chalermsook2020vertex},  to construct a terminal edge connectivity sparsifier for a given terminal set. 
For a set of terminal vertices $T\subset V$, a cut is called a \textit{terminal cut} if the cut partitions $T$ into two non-empty sets, otherwise, the cut is called a \textit{non-terminal cut}. 
Terminal cuts for terminal set $T$ that are also $c$-cuts are called \emph{$(T, c)$-cuts}. A $(T, c)$-cut that partitions $T$ into $T'$ and $T\cut T'$ is called a $(T', T\cut T', c)$-cut.

\begin{Definition}[cut containment set]\label{def:cut_containment_set}
Let $G = (V, E)$ be a graph, $T\subset V$ be a terminal set, and $c$ be a positive integer. 
A set of edges $CC$ is a $c$-cut containment set with respect to $T$ if $CC$ corresponds to intercluster edges of a vertex partition of $G$, and
for every bipartition $(T', T\cut T')$ of $T$ such that $G$ has a $(T', T\cut T', c)$-cut, there exists a minimum $(T', T\cut T', c)$-cut $C$ such that all the edges in the cut-set of $C$ are in $CC$.
\end{Definition}

The terminal edge connectivity sparsifier also makes use of the contraction technique used in~\cite{henzinger1997fully, holm2001poly, nanongkai2017dynamic}. 
Given an unweighted forest $F = (V, E)$ and a set of terminals $K \subset V$, 
the contraction of $F$ with respect to $K$, denoted by $\contract_K(F)$, 
is obtained by repeatedly deleting degree-1 vertices not in $K$ and shortcuting degree-2 vertices not in $K$.
$\contract_K(F)$ may have more vertices than just terminals, but $\contract_K(F)$ 
is the minimal unweighted forest (in terms of the number of vertices and edges) satisfying the following two conditions:
\begin{enumerate}
    \item All the vertices in $K$ are vertices of $\contract_K(F)$ such that any two vertices in $K$ are connected in $\contract_K(F)$ iff they are connected in $F$. 
    \item Each edge of $\contract_K(F)$ corresponds to a path of $F$ that is part of a path connecting two vertices of $K$, and for any two edges of $\contract_K(F)$, their corresponding paths of $F$ are edge disjoint.
\end{enumerate}
As proved in~\cite{henzinger1997fully}, 
$\contract_K(F)$ contains $O(|K|)$ vertices and edges. 
See Appendix~\ref{sec:contraction} for more details of the contraction technique. We remark that $\contract_K(F)$ can contain more vertices than $K$.

In~\cite{js20}, 
given a graph $G$, a terminal set $T$, a $c$-cut containment set $CC$ for $G$ with respect to $T$ for some positive integer $c$, and a positive integer $\gamma$ greater than $c$, 
an unweighted spanning forest $F$ of $G\setminus CC$,
the terminal $c$-edge connectivity sparsifier of $G$ with respect to $T$, denoted as $\sparsifier_{c, \gamma}(G, T, CC, F)$, is 
the multigraph with the vertex set of $\contract_{T\cup \endpoints(CC)}(F)$ as vertices, 
 the union of $CC$ and the edge set of $\contract_{T\cup \endpoints(CC)}(F)$ as edges, and edge multiplicities defined as follow: for the edge with $x$ and $y$ as two endpoints in $\sparsifier_{c, \gamma}(G, T, CC, F)$

 \begin{itemize}
     \item if there is an edge $(x, y, \alpha) \in CC$ for some positive integer $\alpha$, the multiplicity of the edge with $x$ and $y$ as two endpoints in $\sparsifier_{c, \gamma}(G, T, CC, F)$ is $\alpha$;
     \item otherwise the multiplicity of the edge with $x$ and $y$ as two endpoints in $\sparsifier_{c, \gamma}(G, T, CC, F)$ is $\gamma$.
 \end{itemize}
\cite{js20} showed that $\sparsifier_{c, \gamma}(G, T, CC, F)$ has $O(|CC| + |T|)$ vertices and $O(|CC|+ |T|)$ edges with different endpoints, and the $c$-edge connectivity in $G$ for any pair of vertices within $T$ is the same as that of $\sparsifier_{c, \gamma}(G, T, CC, F)$.

For the purpose of efficient updates, 
the terminal $c$-edge connectivity sparsifier is also defined with respect to vertex partitions. 
Given a graph $G$ and a vertex partition $\calP$ of $G$, the terminal set for $G$ and $\calP$, denoted as $T_{G, \calP}$, is defined as $\endpoints(\partial_G(\calP))$.
For a positive integer $c$, 
an edge subset of $G$ is a \emph{$c$-cut containment set for $G$ and $\calP$} if the edge set is the union of $\partial_G(\calP)$ and the $c$-cut containment sets of $G[P]$ with $T_{G, \calP}\cap P$ as terminals for all $P\in\calP$. 

Given a graph $G$, a vertex partition $\calP$ of $G$, positive integers $c, \gamma$ such that $\gamma > c$, a $c$-cut containment set $CC$ for $G$ and $\calP$,  and an unweighted spanning forest $F$ of $G\setminus CC$,
the terminal $c$-edge connectivity sparsifier of $G$ with respect to $\calP$, $CC$ and $F$ is constructed as follows: start from $G$, for each $P \in \calP$, replace $G[P]$ in $G$ by 
\[\sparsifier_{c, \gamma}(G[P], T_{G, \calP}\cap P, CC\cap E(G[P]), F[P]).\]
Recall here that $T_{G, \calP}=\endpoints(\partial_G(\calP))$.
We use $\sparsifier_{c, \gamma}(G, \calP, CC, F)$ to denote the resulting sparsifier.


\begin{lemma}[\cite{js20}]
Let $G = (V, E)$ be a graph, $\calP$ be a vertex partition of $G$, $c$ and $\gamma$ be positive integers such that $\gamma > c$, $CC$ be a $c$-cut containment set for $G$ and $\calP$, and $F$ be an unweighted spanning forest of $G\setminus CC$.
$\sparsifier_{c, \gamma}(G, \calP, CC, F)$ is a terminal $c$-edge connectivity sparsifier of $G$ with respect to $T_{G, \calP}$.
\end{lemma}

We remark that the edge connectivity parameter gets smaller after the sparsifier gets updated with respect to graph updates (see Lemma~\ref{lem:old_main}). In order to answer $c$-edge connectivity queries,
we set the edge connectivity parameter to be large enough so that as the sparsifier updated we can still use the sparsifier to answer $c$-edge connectivity queries, as 
a terminal $c'$-edge connectivity sparsifier is also a terminal $c$-edge connectivity sparsifier for any $c' > c$.


\subsection{One-Level Terminal Edge Connectivity Data Structure}

We summarize the one-level terminal edge connectivity data structure from~\cite{js20}, and the algorithms for the one-level data structure.


For the convenience of constructing and updating terminal edge connectivity sparsifiers with respect to a graph, we suppose that an input graph and its terminal edge connectivity sparsifier are maintained by a \emph{one-level terminal edge connectivity data structure}.
A one-level terminal edge connectivity data structure contains an input graph $G$, a vertex partition $\calP$, a cut containment set $CC$, an unweighted spanning forest $F$ of $G\setminus CC$, the corresponding terminal edge connectivity sparsifier $\sparsifier_{c, \gamma}(G, \calP, CC, F)$ for some positive integers $c$ and $\gamma$, and other auxiliary data structure so that when updates of $G$ occur, $\calP$ and $CC$, $F$, and their corresponding terminal edge connectivity sparsifier can be updated efficiently.

We suppose that $\calP$ is a $\phi$-expander decomposition of $G$ for some $\phi = 1 / n^{o(1)}$, as the terminal edge connectivity sparsifier can only be efficiently constructed and updated for such a $\calP$ by~\cite{js20} (see Lemma~\ref{lem:old_main_initialization} and Lemma~\ref{lem:old_main} for detailed discussion).
We also assume that the data structure contains $\contract_{\endpoints(CC)}(F)$ and the terminal set $T_{G, \calP}$ such that vertices of $T_{G, \calP}$ are stored by the cluster they belong to so that given the index of a cluster $P\in\calP$, it takes $O(1)$ time to return the access of $T_{G, \calP} \cap P$.

Thus, the one-level terminal edge connectivity data structure has three parameters $c$, $\phi$ and $\gamma$, where $c$ is the edge connectivity parameter for the cut containment set, $\phi$ is the conductance parameter for the expander decomposition, and $\gamma$ is the parameter to construct the terminal edge connectivity sparsifier satisfying $\gamma > c$.
Throughout the paper,
$\gamma$ is a fixed parameter, and the other two parameters, $c$ and $\phi$, get smaller after applying the updates to the input graph (see Lemma~\ref{lem:old_main}).

We use $(G, \calP, CC, F)$ to denote the one-level terminal edge connectivity data structure with respect to given $G, \calP$, $CC$, and $F$, which consists of $\sparsifier_{c, \gamma}(G, \calP, CC, F)$.\\


Now we review the initialization and update algorithms for the terminal edge connectivity data structure. 

\SetKwFunction{FTerminalInitialization}{Edge-Connectivity-One-Level-Initialization}

\SetKwFunction{FTerminalUpdatae}{Edge-Connectivity-One-Level-Update}

\def\factor{2^{O(\log^{1/3} n \log^{2/3} \log n)}}

\begin{lemma}[\cite{js20}]\label{lem:old_main_initialization}
Given a graph $G$, two integers $c$ and $\gamma$ such that $\gamma > c$, and a conductance parameter $0 < \phi < 1$, 
there is an algorithm \FTerminalInitialization to initialize a one-level terminal edge connectivity data structure $(G, \calP, CC, F)$ with parameters $c, \phi$ and $\gamma$ in time $m^{1+o(1)} \cdot c^{O(c)} / \mathrm{poly}(\phi)$ such that $\sparsifier_{c, \gamma}(G, \calP, CC, F)$ contains at most \[m \cdot \phi \cdot \factor\] vertices and edges. 
\end{lemma}

For the updates of the one-level data structure, \cite{js20} presented an offline update algorithm for the one-level terminal edge connectivity data structure, i.e., given an update sequence of the graph, the algorithm update the one-level terminal edge connectivity data structure according to the given update sequence of the graph. 
The update sequence consists of the multigraph update operations defined in Section~\ref{sec:preliminary}. 
Let $\updateseq$ be an update sequence. We use $|\updateseq|$ to denote the length of the update sequence.

The lemma below for the update algorithm shows that for a one-level terminal edge connectivity data structure $(G, \calP, CC, F)$, 
given an update sequence $\updateseq$ of $G$, the one-level terminal edge connectivity data structure can be updated in time linear with respect to the length of the update sequence. 

\begin{lemma}[\cite{js20}]\label{lem:old_main}

Given a one-level terminal edge connectivity data structure $(G, \calP, CC, F)$ with parameters $c^2 + 2c, \phi$ and $\gamma$ satisfying $\gamma > c^2 + 2c$, and an update sequence $\updateseq$ of $G$ which contains vertex/edge insertions/deletions, 
Algorithm 
\FTerminalUpdatae with running time $O(|\updateseq|)\cdot O(c/\phi)^{O(c^2)}$ updates the data structure to $(G', \calP', CC', F')$ with parameters $c, \phi/2^{O(\log^{1/3}n\log\log n)}$ and $\gamma$, and outputs an update sequence $\updateseq'$ and a vertex set $S$ satisfying the following properties:
\begin{enumerate}
    \item $G'$ is the resulted graph of applying $\updateseq$ to $G$.
    \item $\calP'$ is a $\phi/2^{O(\log^{1/3}n\log\log n)}$-expander decomposition of $G'$. Furthermore, each cluster $P'\in\mathcal P'$ either contains only a single vertex or is a subset of a cluster $P\in \mathcal P$ such that $G'[P'] = G[P']$.
    \item $CC'$ is a $c$-cut containment set for $G'$ and $\calP'$. 
    \item $F'$ is an unweighted spanning forest of $G'\setminus CC'$.     
    \item The sparsifier maintained in the data structure is $\sparsifier_{c, \gamma}(G', \calP', CC', F')$.
    \item $\updateseq'$ updates $\sparsifier_{c^2 + 2c, \gamma}(G, \calP, CC, F)$ to $\sparsifier_{c, \gamma}(G', \calP', CC', F')$. The length of $\updateseq'$ is at most $\abs{\updateseq}(10c)^{O(c)}$.
    \item $S = (\endpoints(\partial_G(\calP)) \setminus \endpoints(\partial_{G'}(\calP'))) \cup (\endpoints(\partial_{G'}(\calP')) \setminus \endpoints(\partial_G(\calP)))$ such that $|S| = O(|\updateseq|)$.
\end{enumerate}
\end{lemma}

\subsection{Multi-Level Terminal Edge Connectivity Data Structure}\label{sec:multi-s-t-ds}

Now we give a brief summary of the multi-level terminal edge connectivity data structure.
The motivation of the multi-level data structure is to apply the terminal edge connectivity sparsifier iteratively so that the size of the sparsifiers reduce gradually, and 
the edge connectivity is easy to compute in the sparsifier at the end.

A multi-level terminal edge connectivity data structure has three parameters $c, \phi$, and  $\gamma$,  where $c$ is the edge connectivity parameter, $\phi$ is the conductance parameter, and $\gamma$ is the parameter used to construct the sparsifiers for all the levels. 
A multi-level terminal edge connectivity data structure with parameters $c, \phi$, and $\gamma$
is a set of one-level data structures $\set{(G^{(i)}, \mathcal P^{(i)}, CC^{(i)}, F^{(i)})}_{i=0}^\ell$ satisfying the following conditions:
\begin{enumerate}
    \item[(c1)] $(G^{(i)}, \calP^{(i)}, CC^{(i)}, F^{(i)})$ is a one-level terminal edge connectivity data structure with parameters $c, \phi$  and $\gamma$. 
    \item[(c2)]  $G^{(0)}$ is the input graph, and $G^{(i)} = \sparsifier_{c, \gamma} (G^{(i-1)}, \mathcal P^{(i-1)}, CC^{(i-1)}, F^{(i)})$ for each $0 < i \leq \ell$.

    \item[(c3)] $\calP^{(\ell)}$ contains only one cluster.
\end{enumerate}

To initialize a multi-level data structure for a given graph, we apply the one-level initialization algorithm iteratively until the vertex partition contains only one cluster. 
And given an update sequence $\updateseq$ for graph $G$, a multi-level sparsifier $\set{(G^{(i)}, \mathcal P^{(i)}, CC^{(i)}, F^{(i)})}_{i=0}^\ell$ of $G$ is updated as follows: Suppose $\updateseq^{(0)} = \updateseq$. Starting from $i = 0$ to $\ell$, run the one-level update algorithm on the one-level terminal edge connectivity data structure $(G^{(i)}, \mathcal P^{(i)},  CC^{(i)}, F^{(i)})$ with update sequence $\updateseq^{(i)}$, and let $\updateseq^{(i + 1)}$ be the update sequence outputted by the one-level update algorithm. 
If at some level, the update sequence contains  too many updates so that the update algorithm is slower than recompute the sparsifier, then use the initialization algorithm to reconstruct the sparsifiers iteratively for the remaining levels. 


The fully dynamic algorithm is obtained by applying a framework proposed in~\cite{nanongkai2017dynamic} that converts a batched dynamic update algorithm 
to a fully dynamic algorithm. 
See Appendix~\ref{sec:online_batch_general} for more details of the framework proposed in~\cite{nanongkai2017dynamic}.

\begin{theorem}[Theorem 9.3 and 9.4 from \cite{js20}]\label{thm:js21main}
For any $c = (\log n)^{o(1)}$, 
there is a fully dynamic algorithm which maintains a set of multi-level terminal edge connectivity data structure such that each maintained terminal edge connectivity data structure has $O(\log^{1/10} n)$ levels all the time, and after processing each update, the algorithm provides the access to one of the maintained multi-level terminal edge connectivity data structure   $\set{(G^{(i)}, \mathcal P^{(i)}, CC^{(i)}, F^{(i)})}_{i=0}^\ell$ for the up-to-date graph with parameters $c, \phi, \gamma$ satisfying the following conditions:
\[\phi = 1/n^{o(1)}, \text{ and } \gamma > c. \]
The initialization time of the algorithm is $m^{1 + o(1)}$, and the update time of the algorithm is $n^{o(1)}$ per update.
\end{theorem}

\section{Minimum \texorpdfstring{$c$}{c}-Cut Data Structure}\label{sec:min-c-cut}
We aim to maintain a minimum $c$-cut data structure that allows to compute the minimum $c$-cut of the input graph efficiently. 
In this section, we present our minimum $c$-cut data structure. 
In Section~\ref{sec:init_update}, we present the initialization and update algorithms for the minimum $c$-cut data structure. 

\subsection{The Characterization of Minimum \texorpdfstring{$c$}{c}-cuts}

The high level intuition is to make use of the terminal edge connectivity sparsifier for $c$-edge connectivity described in Section~\ref{sec:js20_review} to speed up the process of finding the minimum $c$-cut if at least one of the minimum $c$-cuts of the input graph is captured by the sparsifier, as the terminal edge connectivity sparsifier is smaller than the input graph in terms of the number of vertices and edges. 
If no minimum $c$-cut is captured by the terminal edge connectivity sparsifier, we need to characterize these minimum $c$-cuts, and devise a new approach to find these cuts efficiently.

The two lemmas below show that the $c$-edge connectivity sparsifier does not only preserve the $c$-edge connectivity between any pair of terminals, but also allows us to reconstruct the $c$-terminal cuts. 

\begin{lemma}\label{lem:sparsifier_terminal_cut_pre}
Let $G$ be a graph, $T$ be a set of terminals, $CC$ be a $c$-cut containment set for $G$ for some positive integer $c$, and $F$ be an unweighted spanning forest of $G\setminus CC$. 
The following properties hold for $\sparsifier_{c, \gamma}(G, T, CC, F)$ for any $\gamma > c$:
\begin{enumerate}
    \item The cut-set of any $c$-cut in $\sparsifier_{c, \gamma}(G, T, CC, F)$ is the cut-set of a cut with the same cut size in $G$. 
    \item If there is a $c$-cut $C$ in $G$ such that $C$ partitions $T$ into two non-empty subsets $T'$ and $T\setminus T'$, then there is a cut $C'$ of size smaller than or equal to the size of $C$ in $\sparsifier_{c, \gamma}(G, T, CC, F)$ that partitions $T$ into $T'$ and $T\setminus T'$. 
\end{enumerate}
\end{lemma}
\begin{proof}
By Definition~\ref{def:cut_containment_set}, $CC$ is the set of intercluster edges of a vertex partition of $G$. Let $\calR$ denote this vertex partition. 
By the construction of $\sparsifier_{c, \gamma}(G, T, CC, F)$, $\sparsifier_{c, \gamma}(G, T, CC, F)$ contains all the edges of $CC$. Comparing $G$ and $\sparsifier_{c, \gamma}(G, T, CC, F)$, 
for each $R\in \calR$, 
$G[R]$ in $G$ is replaced by a tree in  $\sparsifier_{c, \gamma}(G, T, CC, F)$ such that the edge multiplicity of each edge in the tree is greater than $c$. 
Furthermore, the vertices of $\sparsifier_{c, \gamma}(G, T, CC, F)$, which contain all the endpoints of edges in $CC$, are also vertices of $G$.

With these properties of $\sparsifier_{c, \gamma}(G, T, CC, F)$, any $c$-cut of $\sparsifier_{c, \gamma}(G, T, CC, F)$ only contains edges from $CC$.
Also,
for any two vertices $u$ and $v$ in $\sparsifier_{c, \gamma}(G, T, CC, F)$, 
if there is a $R\in \calR$ containing both $u$ and $v$,
$u$ and $v$ are in the same side of any $c$-cut of $\sparsifier_{c, \gamma}(G, T, CC, F)$. 
Let $C' = (V_1', V_2')$ be a $c$-cut of $\sparsifier_{c, \gamma}(G, T, CC, F)$. 
We construct cut $C = (V_1, V_2)$ of $G$ as follows: for any vertex $v$ of $G$, if there is a vertex $x\in V_1'$ such that there is a $R\in \calR$ containing both $v$ and $x$, then $v$ is in $V_1$, otherwise, $v$ is in $V_2$. 
By the construction of $C$, we have the following properties:
\begin{enumerate}
    \item For any two vertices $x$ and $y$ in $\sparsifier_{c, \gamma}(G, T, CC, F)$, if $x$ and $y$ are in the same side of $C'$, $x$ and $y$ are in the same side of $C$.
    \item For any two vertices $x$ and $y$ of $G$, if there is a $R\in \calR$ such that both $x$ and $y$ are in $R$, then $x$ and $y$ are in the same side of $C$.
\end{enumerate}
Since every vertex in $\endpoints(\partial_G(\calR))$ is an endpoint of some edge in $CC$, 
the cut-set of $C$ is the same as the cut-set of $C'$. Hence, the first property holds.

Now we prove the second property. 
By Definition~\ref{def:cut_containment_set}, there is a set of edges $W \subset CC$ such that $W$ is the cut-set of a $(T', T\setminus T', c)$-cut $C_1 = (V_1, V_2)$ of $G$ such that the cut size of $C_1$ is smaller than or equal to $C$. 
We construct a cut $C' = (V_1', V_2')$ of $\sparsifier_{c, \gamma}(G, T, CC, F)$ as follows: for any vertex $v$ of $\sparsifier_{c, \gamma}(G, T, CC, F)$, if there is a vertex $x\in V_1$ such that there is a $R\in \calR$ containing both $v$ and $x$, then $v$ is in $V_1'$, otherwise, $v$ is in $V_2'$. 
By the construction of $C'$, we have the following properties:
\begin{enumerate}
    \item For any two vertices $x$ and $y$ in $\sparsifier_{c, \gamma(G, T, CC, F)}$, if there is a $R\in \calR$ such that both $x$ and $y$ are in $R$, then $x$ and $y$ are in the same side of $C'$.
    \item For any two vertices $x$ and $y$ in $\sparsifier_{c, \gamma(G, T, CC, F)}$, if $x$ and $y$ are in the same side of $C$ for $G$, then $x$ and $y$ are in the same side of $C'$ for $\sparsifier_{c, \gamma(G, T, CC, F)}$.
\end{enumerate}
Thus, the cut-set of $C'$ is the same as the cut-set of $C_1$, and then the second property holds.
\end{proof}


\begin{lemma}\label{lem:sparsifier_terminal_cut}
Let $G$ be a graph, $\calP$ be a vertex partition of $G$, $CC$ be a $c$-cut containment set for $G$ and $\calP$ with some positive integer $c$, and $F$ be an unweighted spanning forest of $G\setminus CC$. 
The following properties hold for $\sparsifier_{c, \gamma}(G, \calP, CC, F)$ for any $\gamma > c$:
\begin{enumerate}
    \item The cut-set of any $c$-cut in $\sparsifier_{c, \gamma}(G, \calP, CC, F)$ is the cut-set of a cut with the same size in $G$. 
    \item If there is a $c$-cut $C$ in $G$ such that $C$  partitions terminals $T_{G, \calP}$ into two non-empty subsets $T'$ and $T_{G, \calP} \setminus T'$, then there is a cut $C'$ of size smaller than or equal to $C$ in the sparsifier $\sparsifier_{c, \gamma}(G, \calP, CC, F)$ such that $C'$ partitions $T_{G, \calP}$ into $T'$ and  $T_{G, \calP} \setminus T'$.
\end{enumerate}
\end{lemma}

\begin{proof}
By the definition of $CC$, $CC$ is the union of $\partial_G(\calP)$ and $c$-cut containment sets of $G[P]$ with $T_{G, \calP}\cap P$ as terminals for all $P\in\calP$. By Definition~\ref{def:cut_containment_set}, $CC$ is 
the set of intercluster edges of a vertex partition of $G$. Let $\calR$ denote this vertex partition. 
$\calR$ is a refinement of $\calP$, i.e., for each $R\in\calR$, there is a $P\in\calP$ such that $R$ is a subset of $P$.

Recall that $\sparsifier_{c, \gamma}(G, \calP, CC, F)$  is obtained by replacing each induced subgraph $G[P]$ in $G$ by \[\sparsifier_{c, \gamma}(G[P], T_{G, \calP}\cap P, CC\cap E(G[P]))\] for all the $P \in \calP$.
By the construction of $\sparsifier_{c, \gamma}(G[P], T\cap P, CC\cap E(G[P]))$, the vertices of $\sparsifier_{c, \gamma}(G, \calP, CC, F)$, which contain all the endpoints of edges in $CC$, are also vertices of $G$.

Hence, any $c$-cut of $\sparsifier_{c, \gamma}(G, \calP, CC, F)$ only contains edges from $CC$. And for any pair of vertices in $\sparsifier_{c, \gamma}(G, \calP, CC, F)$ that are contained by a vertex set of $\calR$, 
they are in the same side of any $c$-cut of $\sparsifier_{c, \gamma}(G, \calP, CC, F)$. 

To prove the first property, let $C' = (V_1', V_2')$ be an arbitrary $c$-cut of $\sparsifier_{c, \gamma}(G, \calP, CC, F)$. 
We construct a cut $C = (V_1, V_2)$ of $G$ as follows: for any vertex $v$ of $G$, if there is a vertex $x\in V_1'$ such that there is a $R\in \calR$ containing both $v$ and $x$, then $v$ is in $V_1$, otherwise, $v$ is in $V_2$. 
Since vertices in $G$ from the same cluster of $\calR$ are in the same side of $C$, 
the cut-set of $C$ is a subset of $CC$. 
By the construction of $C$, two endpoints of any edge in $CC$ are on the same side of $C$ only if the two endpoints are on the same side of $C'$ for $\sparsifier_{c, \gamma}(G, \calP, CC, F)$. 
Hence, 
The cut-set of $C$ is the same as the cut-set of $C'$, and thus the first property holds.

Now we prove the second property. 
Let $S$ denote the cut-set of $G$. 
Consider an arbitrary $P_1 \in \calP$. Let $S_{P_1} = S \cap E(G[P_1])$. 
If $S_{P_1}$ is the cut-set of a terminal cut $C_{P_1}$ for $G[P_1]$ with $T_{G, \calP}\cap P_1$ as terminals, then by Lemma~\ref{lem:sparsifier_terminal_cut_pre}, there is a set of edges $S_{P_1}' \subset CC$ satisfying the following conditions:
\begin{itemize}
    \item  $S_{P_1}'$ is the cut-set of a cut $C_{P_1}'$ for $G[P_1]$.
    \item $C_{P_1}$ and $C_{P_1}'$ have the same partition on $T_{G, \calP}\cap P_1$. 
    \item The cut size of $C_{P_1}'$ is smaller than or equal to $C_{P_1}$. 
\end{itemize}
If $S_{P_1}$ is not the cut-set of a terminal cut for $G[P_1]$ with $T_{G, \calP}\cap P_1$ as terminals, then we let $S_{P_1}'$ be an empty set.

Let $G_1$ be the graph obtained by using $\sparsifier_{c, \gamma}(G[P_1], T_{G, \calP}\cap P_1, CC\cap E(G[P_1]))$ to replace $G[P_1]$ in graph $G$. 
$(S \setminus S_{P_1})\cup S_{P_1}'$ is the cut-set of a cut $C_1$ for graph $G_1$ satisfying the following conditions:
\begin{itemize}
    \item The cut size of $C_1$ is smaller than or equal to that of $C$.
    \item $C$ and $C_1$ have the same partition on $T_{G, \calP}$ for $G$ and $G_1$ respectively.
    \item The edges of $C_1$'s cut-set that are in $G[P_1]$ are in $CC$. 
\end{itemize} 

If the above process is repeated for all the other clusters in $\calP$, then at the end, we obtain a cut $C'$ for $\sparsifier_{c, \gamma}(G, \calP, CC, F)$ satisfying the following conditions:
\begin{itemize}
    \item The cut size of $C'$ is smaller than or equal to that of $C$.
    \item $C$ and $C'$ have the same partition of $T_{G, \calP}$ for $G$ and $\sparsifier_{c, \gamma}(G, \calP, CC, F)$ respectively.
    \item The cut-set of $C'$ are in $CC$. 
\end{itemize} 
Thus, the second property of the lemma holds. 
\end{proof}




By Lemma~\ref{lem:sparsifier_terminal_cut}, if there is a minimum cut of the input graph being a terminal cut, then the terminal cut can be obtained from the terminal edge connectivity sparsifier.
Otherwise, no minimum cut is a terminal cut, the terminal edge connectivity sparsifier is not helpful. Hence, we need to devise a new approach for the case that all the minimum $c$-cuts of the input graph are not terminal cuts.
The following lemma characterizes the minimum cuts of a graph with respect to a vertex partition. 


\begin{lemma}\label{lem:char}
Let $G$ be a graph, and $\mathcal P$ be a vertex partition of $G$. 
Let $C$ be a minimum cut on $G$. $C$ is either a terminal cut with $T_{G, \calP}$ as terminals or a cut with cut-set in $G[P]$ for some $P\in\mathcal P$ such that all the terminal vertices in $T_{G, \calP}$ are in the same side of $C$.
\end{lemma}

\begin{proof}
If the cut-set of $C$ contains an intercluster edge with respect to $\calP$, then $C$ is a terminal cut, because the two endpoints of the intercluster edge are terminals, and belong to different sides of the cut. 

On the other hand, let $P$ be a cluster of $\calP$ such that the edges in the cut-set of $C$ that are also in $G[P]$ form a terminal cut for $G[P]$ with $T_{G, \calP}\cap P$ as terminals. 
Since the two endpoints of any edge in the cut-set of $C$ are on different sides of $C$, 
vertices of $T_{G, \calP}\cap P$ are not in the same side of $C$ for $G$, and thus $C$ is a terminal cut.



Hence, if $C$ is not a terminal cut of $G$, then every edge in the cut-set of $C$ is an intracluster edge with respect to $\calP$, and 
for each $P\in \calP$, one of the following two conditions hold:
\begin{enumerate}
    \item The cut-set of $C$ does not contain any edge in $G[P]$.
    \item The edges in the cut-set of $C$ that are also in $G[P]$ form a non-terminal cut of $G[P]$ with $T_{G, \calP} \cap P$ as terminals. 
\end{enumerate}
For any $P \in \calP$, let $C'$ be an arbitrary cut of $G[P]$ such that the $T_{G, \calP}\cap P$ are in the same side of $C'$. 
The cut-set of $C'$ is also the cut-set of a cut in graph $G$. 
Hence, if $C$ is a minimum cut of $G$ that is a non-terminal cut with $T_{G, \calP}$ as terminals, then 
there is a $P\in \calP$ such that $G[P]$ contains all the edges in the cut-set of $C$, and the cut-set of $C$ forms a non-terminal cut of $G[P]$ with $T_{G, \calP}\cap P$ as terminals. 
Since any path from a vertex in $P$ to a vertex not in $P$ must contain a vertex in $T_{G, \calP} \cap P$, 
$C$ is a cut such that all the vertices of $T_{G, \calP}$ are in the same side of $C$.
\end{proof}

Based on Lemma~\ref{lem:sparsifier_terminal_cut} and Lemma~\ref{lem:char}, 
since terminal cuts of size at most $c$ are captured by the sparsifier, 
we use the terminal sparsifier to find the minimum terminal $c$-cuts, and need to devise a new approach to find the minimum non-terminal $c$-cuts.



\subsection{Local Cut}
We define local cut, the information we need to maintain in our global minimum cut data structure in order to find the minimum cut efficiently.


Recall that our goal is to find the minimum non-terminal $c$-cuts with $T_{G, \calP}$ as terminals for a graph $G$ and a vertex partition $\calP$. 
By Lemma~\ref{lem:char}, every minimum non-terminal $c$-cut must have the cut-set belong to a cluster in $\calP$.
Hence, we only need to focus on non-terminal cuts with the cut-set in a cluster.
On the other hand, note that each cluster is supposed to be a $1/n^{o(1)}$-expander as Theorem~\ref{thm:js21main}.
By the definition of expander, every cut of size at most $c$ has a side of volume at most $O(c\cdot n^{o(1)})$.  We use the following definition to describe non-terminal cuts in an expander.

\begin{Definition}\label{def:local_cut}
Let $G = (V, E)$ be a graph with $T$ as terminals, and $\alpha, c$ be two positive integers. We say a cut $C = (U, V\cut U)$ is an $(\alpha, c)$-local cut in $G$ with respect to $T$ if all of the following conditions hold
\begin{itemize}
    \item The size of cut $C$ is at most $c$;
    \item Either $U$ or $V\cut U$ has volume at most $\alpha$ in $G$;
    \item $T\subset U$ or $T\subset V\cut U$.
\end{itemize}
A cut $C$ is a minimum $(\alpha, c)$-local cut in $G$ with respect to $T$ if $C$ is a $(\alpha, c)$-local cut in $G$ with respect to $T$ and there is no other $(\alpha, c)$-local cuts in $G$ with respect to $T$ with the size smaller than $C$.
\end{Definition}

The following lemma shows that if there exists a $c$-cut in a graph, then there is a minimum cut that is either a minimum local cut with appropriate parameters or a cut with cut-set that can be obtained from the terminal edge connectivity sparsifier. 

\begin{lemma}\label{lem:min_cut_char}
Let $G$ be a graph, $\calP$ be a $\phi$-expander decomposition of $G$, $CC$ be a $c$-cut containment set for $G$ and $\calP$ with some fixed positive integer $c$, $F$ be an unweighted spanning forest of $G\setminus CC$, and $\sparsifier_{c, \gamma}(G, \calP, CC, F)$ be a terminal edge connectivity sparsifier for $G$ and $\calP$ with parameter $\gamma > c$. If the minimum cut of $G$ has size at most $c$, then there is a minimum cut $C$ of $G$ satisfies one of the following two conditions:
\begin{enumerate}
    \item 
The cut set of $C$ is also the cut set of a minimum cut of  $\sparsifier_{c, \gamma}(G, \calP, CC, F)$.
    \item $C$ is a minimum $(\alpha, c)$-local cut in $G[P]$ for some $P\in\mathcal P$ with respect to $\endpoints(\partial_G(\mathcal P))\cap P$, and any $\alpha \geq c / \phi$. 
\end{enumerate}
\end{lemma}
\begin{proof}
By Lemma~\ref{lem:char}, any minimum cut of $G$ is either a terminal cut with $\endpoints(\partial_G(\calP))$ as terminals, or a cut within $G[P]$ for some $P \in \calP$ such that all the terminal vertices in $G[P]$ are in the same side of the cut. 

If there is a minimum $c$-cut of $G$ that is also a terminal cut, by the construction of the terminal cut sparsifier, the cut-set of a minimum cut of $\sparsifier_{c, \gamma}(G, \calP, CC, F)$ is also the cut-set of a minimum cut of $G$. 

If there is a minimum $c$-cut of $G$ which is a cut within a $G[P]$ for some $P \in \calP$ such that all the terminal vertices in $G[P]$ are on the same side of the cut, then one side of the cut is of volume at most $c / \phi$ since $G[P]$ is a $\phi$-expander, by Definition \ref{def:local_cut}, it is a $(\alpha, c)$-local cut in $G[P]$ with respect to $\endpoints(\partial_G(\mathcal P))\cap P$. Since it is a minimum cut, it is a minimum $(\alpha, c)$-local cut in $G[P]$ because every  $(\alpha, c)$-local cut is also a cut for the entire graph. 
\end{proof}

\def\parametertimes{\xi}
\def\parametertimesub{\zeta}
\def\parameterlength{w}
\def\dsinitialize{\textsc{DS-Initialize}}
\def\dsupdate{\textsc{DS-Update}}
\def\rti{t_{\mathrm{initialization}}}
\def\rtu{t_{\mathrm{update}}}
\def\parametertimes{\xi}
\def\parameterlength{w}
\def\parametertimesub{\zeta}

\newcommand{\condname}[1]{(#1)-small}

\subsection{The Global Min-Cut Sparsifier Data Structure}\label{sec:init_update}

Lemma \ref{lem:char} states that a global minimum cut is either a terminal cut or a local cut. Given a graph $G$, in order to keep track of the global minimum cut, we use a terminal cut data structure to keep track of all terminal cuts, and try including an additional data structure to find a minimum local cut. In order to keep track of all local cuts in $G$, we need help from auxiliary graphs.

\vspace{.15cm}\noindent \textbf{Auxiliary Graph and Partition \ \ } 
We first construct an auxiliary graph based on $G$ by connecting all terminal vertices. The auxiliary graph makes sure that in a local cut, the side containing all terminal vertices is always connected. 
This makes it easier to find minimum local cuts in the auxiliary graph.

\begin{Definition}\label{def:aux-graph}
 For a graph $G = (V, E)$ and a partition $\mathcal P$ of $V$,
 the auxiliary graph of $G$ with respect to $\mathcal P$, denoted as $G_*$,  is obtained from $G$ as follows: 
 For each $P\in\mathcal P$, add a special terminal vertex $t_P$ and connect $t_P$ to every terminal vertex in $P$. 
 
 The auxiliary partition $\calP_*$ is defined as  $\{P_* | P\in \mathcal P\}$, where $P_* = P\cup \set {t_P}$. 
The terminals of $G_*$ with respect to $\calP_*$ are the terminals of $G$ with respect to $\calP$ union all the special terminals.
 \end{Definition}

The following lemma shows how the auxiliary graph is able to help us find minimum local cuts in the original graph.
\begin{lemma}\label{lem:special-terminal}
Let $G$ be a graph and $\mathcal P$ be a $\phi$-expander decomposition of $G$. 
Let $G_*$ be the auxiliary graph of $G$ with respect to $\mathcal P$. Let $\alpha$ and $c$ be two parameters such that $\alpha\geq \ceil{c/\phi}$.
The following two conditions hold for any $P \in \calP$:
\begin{enumerate}
\item If $(U, P_*\cut U)$ is a minimum $(3\alpha,c)$-local cut in $G_*[P_*]$, then $(U\cap P, (P_*\cut U)\cap P)$ is a minimum $(\alpha, c)$-local cut in $G[P]$. 
\item Suppose $(U', P\cut U')$ is a minimum $(\alpha, c)$-local cut in $G[P]$. Without loss of generality, assume all terminals of $P$ are in $U'$. 
Then $(U'\cup\set{t_P}, P\cut U')$ is a minimum $(3\alpha, c)$-local cut in  $G_*[P_*]$. 
\end{enumerate}
\end{lemma}
\begin{proof}
Let $(U, P_*\cut U)$ be a minimum $(3\alpha,c)$-local cut in $G_*[P_*]$. Whether $t_P\in U$ or $t_P\in P_*\cut U$, since all terminals are on the same side of the cut, removing edges between $t_P$ and other terminals does not affect the number of edges in the cut-set. Therefore, $(U\cap P, (P_*\cut U)\cap P)$ is still a $c$-cut in $G[P]$. By the assumption that $G[P]$ is a $\phi$-expander, one side of $(U\cap P, (P_*\cut U)\cap P)$ has volume at most $\alpha$. Therefore, $(U\cap P, (P_*\cut U)\cap P)$ is a $(\alpha, c)$-local cut in $G[P]$ that has the same size as $(U, P_*\cut U)$.
 $(U\cap P, (P_*\cut U)\cap P)$ must be a minimum $(\alpha, c)$-local cut in $G[P]$ because otherwise $(U, P_*\cut U)$ is not a minimum $(3\alpha,c)$-local cut in $G_*[P_*]$.

Let  $(U', P\cut U')$ be a minimum $(\alpha, c)$-local cut in $G[P]$ such that the terminals are in $U'$. We look at cut $(U'\cup\set{t_P}, P\cut U')$ in $G_*[P_*]$. Since the edges between $t_P$ and other terminals are not in the cut-set, the size of this cut is at most $c$. Since $G[P]$ is a connected graph and the volume of $U'$ in $G[P]$ is at most $\alpha$,
we add at most $\alpha$ new edges to turn $G[P]$ to $G_*[P_*]$. The volume of $U'\cup\set{t_P}$ in $G_*[P_*]$ is at most $3\alpha$. Therefore, $(U'\cup\set{t_P}, P\cut U')$ is a $(3\alpha, c)$-local cut in  $G_*[P_*]$ that has the same cut size as $(U', P\cut U')$. 
In addition, $(U'\cup\set{t_P}, P\cut U')$ must be a minimum $(3\alpha, c)$-local cut in $G_*[P_*]$, because if there exists a smaller $(3\alpha, c)$-cut, then its corresponding cut in $G$ would be an $(\alpha, c)$-local cut which has a size smaller than $(U', P\cut U')$, and thus $(U', P\cut U')$ would not be a minimum $(\alpha,c)$-local cut in $G[P]$.
\end{proof}


As a result of Lemma \ref{lem:special-terminal}, instead of finding minimum local cuts in $G$ with respect to $\calP$, we work on finding the minimum local cuts in $G_*$ with respect to $\calP_*$. The major advantage of using $G_*$ and $\calP_*$ is that we only need to consider cuts with both sides connected when considering the connected components of $G_*[\calP_*]$. 
We note that our approach to finding the minimum cut is different from~\cite{saranurak2021simple}, which also finds the minimum cut of the graph by investigating a vertex partition. 
Our approach is based on finding minimum local cuts, whereas \cite{saranurak2021simple} relies on Gabow's matroid characterization for edge connectivity~\cite{gabow1995matroid}. 




\begin{lemma}\label{lem:special-terminal-simple-cut}
Let $G$ be a graph and $\mathcal P$ be a $\phi$-expander decomposition of $G$. 
Let $G_*$ and $\calP_*$ be the auxiliary graph and partition of $G$ with respect to $\mathcal P$. Let $\alpha$ and $c$ be two parameters such that
$\alpha\geq \ceil{c/\phi}$.
Let $C$ be the cut-set of a minimum $(3\alpha, c)$-local cut of $G_*$ with respect to $\calP_*$, and $P_* \in \calP_*$ be the vertex partition such that $C$ contains only edges in $G_*[P_*]$. 
 The graph obtained by removing  $C$ from $G_*[P_*]$ has only two connected components.

\end{lemma}

\begin{proof}


Let $P$ be the cluster in $\mathcal P$ that corresponds to $P_*\in \mathcal P_*$. Since $\mathcal P$ is a $\phi$-expander decomposition, $G[P]$ is a $\phi$-expander. By construction, $G_*[P_*]$ is obtained by adding a new terminal vertex $t_P$ to $G[P]$ and connecting $t_P$ to existing terminal vertices. For any bipartition $(X_*, Y_*)$ of $P_*$, such that neither $X_*$ nor $Y_*$ is $\set{t_P}$ (we can assume $t_P\in X_*$ without loss of generality), there exists a corresponding bipartition $(X  = X_*\cap P, Y = Y_*\cap P)$ of $P$. Since $G[P]$ is a $\phi$-expander, 
$$\frac{\abs{\partial_{G[P]}(X)}}{\min (\vol_{G[P]}(X), \vol_{G[P]}(Y))}\geq \phi.$$
Suppose among all edges added to construct $G_*[P_*]$ from $G[P]$, $F_1$ is the set of edges that cross bipartition $(X_*, Y_*)$ and $F_2$ is the set of edges that do not cross the bipartition, then
$$\frac{\abs{\partial_{G_*[P_*]}(X_*)}}{\min (\vol_{G_*[P_*]}(X_*), \vol_{G_*[P_*]}(Y_*))} =\frac{\abs{\partial_{G[P]}(X)} + \abs{F_1}}{\min (2\cdot \abs{F_2 }+\vol_{G[P]}(X), \vol_{G[P]}(Y)) }.$$
Since $F_2$ consists only of edges connecting $t_P$ and terminals in $X$, $\abs{F_2}$ is at most $\vol_{G[P]}(X)$.
Therefore, 
\begin{align*}
\frac{\abs{\partial_{G[P]}(X)} + \abs{F_1}}{\min (2\cdot \abs{F_2 }+\vol_{G[P]}(X), \vol_{G[P]}(Y)) } &\geq \frac{\abs{\partial_{G[P]}(X)}+ 0}{\min (3\cdot \vol_{G[P]}(X), \vol_{G[P]}(Y)) }\\&\geq \frac{\abs{\partial_{G[P]}(X)}}{3\cdot \min ( \vol_{G[P]}(X), \vol_{G[P]}(Y)) }\\ & \geq  \phi/3.
\end{align*}

If $X_*$ or $Y_*$ is $\set {t_P}$, then $$\frac{\abs{\partial_{G_*[P_*]}(X_*)}}{\min (\vol_{G_*[P_*]}(X_*), \vol_{G_*[P_*]}(Y_*))}=1 \geq \phi.$$


Hence $G_*[P_*]$ is a $\phi/3$-expander. Since $C$ forms a minimum cut, no proper subset of $C$ can form a local cut in $G_*[P_*]$ (otherwise because $G_*[P_*]$ is a $\phi / 3$-expander, any cut with terminals on one side has a side of size at most $3\alpha$, resulting in a smaller $(3\alpha, c)$-local cut).

We look at the induced subgraphs after removing $C$ from $G_*[P_*]$. Note that since all terminals including $t_P$ are on the same side of the cut and $t_p$ connects to all the other terminals, there can be only one connected component containing all terminals. 

Consider the two subgraphs induced by the two sides of the cut are disjoint. 
Suppose at least one of these two subgraphs contains more than one connected component. Let $X$ be one of these connected components. Since $G_*[P_*]$ is connected, $X$ is connected to the other side of the cut by edges in $C$.
Removing edges connecting $X$ and the other side will result in a proper subset of $C$ that still forms a local cut (equivalent to moving $X$ to the other side; because all terminals are in the same connected component, this resulting cut is still local).  We obtained a contradiction. Thus, both subgraphs induced by the two sides of the cut are connected.
\end{proof}

\vspace{.15cm}\noindent \textbf{One-level Sparsifier Definition \ \ }
Since a global minimum cut is either a terminal cut or a local cut, in order to maintain a global minimum cut, it suffices to maintain a minimum terminal cut and a minimum local cut. Therefore, a one-level global minimum cut data structure has an internal one-level terminal edge connectivity data structure (to keep track of terminal cuts),
and a priority queue $\Lambda_{G_*}$ to store local cuts in $G_*$ (to keep track of local cuts in $G$, as per Lemma \ref{lem:special-terminal}). Here $G_*$ is the auxiliary graph of $G$ with respect to $\mathcal P$. We rigorously define the one level sparsifier data structure below.
\begin{Definition}[One-level Sparsifier]\label{def:lambda}
 Given parameters $c, \phi$, $c'\geq c$, $\gamma\geq c'$,  and $\alpha\geq \ceil{c/\phi}$, a one-level global minimum cut data structure w.r.t. these parameters consists of two parts:
\begin{enumerate}
    \item[(1)] A one-level terminal edge connectivity data structure  $\sparsifier_{c', \gamma}(G, \calP, CC, F)$ (Section \ref{sec:js20_review}), which contains the input graph $G = (V, E)$, a vertex partition $\calP$ of $G$ given by a $\phi$-expander decomposition, and a $c'$-cut containment set $CC$ with respect to $G$ and $\calP$.

    \item[(2)] Assuming $G_*$ is the auxiliary graph of $G$ with respect to $\mathcal P$, a priority queue $\Lambda_{G_*}$ that contains all vertex sets $U$ such that
    \begin{itemize}
    \item There exists a $P\in\mathcal P$ such that $U\subset P_*$, where $P_*$ is the corresponding cluster of $P$ in $G_*$
    \item $U$ induces a connected subgraph in $G_*$
    \item  $(U, P_*\cut U)$ is a $(3\alpha, c)$-local cut in $G_*[P_*]$. The terminal set is $(\endpoints(\partial_G(\mathcal P))\cap P) \cup\set{t_P}$.
\end{itemize}
The elements in $\Lambda_{G_*}$ are ordered by the size of each cut $(U, P_*\cut U)$.



\end{enumerate}
\end{Definition}

 Note that, by storing a vertex subset $U$ of $G_*$, we are essentially storing the cut $(U, P_*\cut U)$. Since $(U, P_*\cut U)$ is a cut in $G_*[P_*]$ with all boundary vertices on the same side, the cut-set of $(U, P_*\cut U)$ induces a cut $(U, V_*\cut U)$ in $G_*$, with the same size.  
 By Lemma \ref{lem:special-terminal-simple-cut}, all minimum $(3\alpha, c)$-local cuts in $G_*[P_*]$ have two connected sides. 
 Therefore, if $\Lambda_{G_*}$ contains all \textit{connected} $U$ such that $(U, P_*\cut U)$ is a $(3\alpha, c)$-local cut, $\Lambda_{G_*}$ contains all possible $U$ such that $(U, P_*\cut U)$ is a minimum $(3\alpha, c)$-local cut.
 We can see $\Lambda_{G_*}$ as a priority queue that allows us to keep track of minimum $(3\alpha, c)$-local cuts in $G_*$.

\begin{Remark}
In the definition of the one-level data structure, we actually store $c'$-cuts in the terminal sparsifier for $c'\geq c$. This is because in the terminal sparsifier update algorithm (Lemma \ref{lem:old_main}), 
the edge connectivity parameter for the sparsifier reduces after applying an update sequence to the input graph. 
However, this does not affect the query result -- if the minimum terminal cut has a size at most $c$, it is also a minimum $c'$-terminal cut. We will refer to these cuts as ``minimum $c$-cuts in $\sparsifier_{c', \gamma}(G, \mathcal P, \cc, F)$''.
\end{Remark}

\begin{Remark}
Because $G_*[P_*]$ is an expander, $\alpha$ can be set to any number greater than $c/\phi$. This does not change the vertex sets stored in $\Lambda_{G_*}$. In fact, in Section \ref{sec:one-level-init}, we will set $\alpha$ to be $\ceil{c'/\phi}$ for a simpler complexity expression.
\end{Remark}

 Below (Lemma \ref{lem:one-level-query}) we show that the highest priority item in the priority queue gives us a minimum local cut in $G_*$, and thus a minimum local cut in $G$ (by Lemma \ref{lem:special-terminal}). The global minimum cut of $G$ can be restored by comparing the size of the minimum terminal cut found in the terminal cut edge connectivity data structure and the minimum local cut found in $\Lambda_{G_*}$.
 
 \begin{lemma}\label{lem:one-level-query}
Let $\sparsifier_{c', \gamma}(G, \mathcal P, CC, F)$ be a one-level terminal edge connectivity sparsifier data structure with parameters $c'$, $\gamma$ such that $c'\geq c$ and $\gamma \geq c'$. 
Let $\Lambda_{G_*}$ be the priority queue defined in Definition \ref{def:lambda}.
Let $C_1$ be the cut in $G$ induced by the cut-set of a minimum $c$-cut in $\sparsifier_{c', \gamma}(G, \mathcal P, CC, F)$ (if one such cut exists in $\sparsifier_{c', \gamma}(G, \mathcal P, CC, F)$); otherwise $C_1 = \perp$.



Let $U_2$ be a vertex set in $\Lambda_{G_*}$ with highest priority. If $U_2$ exists, let $P_*$ be the cluster in $G_*$ that contains $U_2$ and $P$ be the cluster in $\mathcal P$ corresponding to $P_*$. We use $C_2$ to denote cut $(U_2\cap V, V_*\cut U_2\cap V)$. If $U_2$ does not exist, $C_2 = \perp$.

If $ C_1$ and $C_2$ are both $\perp$, then $G$ does not have a $c$-cut. Otherwise, the smaller one between $ C_1$ and $C_2$ is a minimum $c$-cut of $G$.
\end{lemma}
\begin{proof}
By Lemma~\ref{lem:min_cut_char}, any minimum cut of $G$ is either a terminal cut with $\endpoints(\partial_G(\calP))$ as terminals, or a local cut within a $G[P]$ for some $P \in \calP$ with respect to $\endpoints(\partial_G(\mathcal P))\cap P$, $\alpha \geq \ceil{c / \phi}$. By Lemma \ref{lem:sparsifier_terminal_cut}, $C_1$ is a minimum terminal cut in $G$. 

By definition of $\Lambda_{G_*}$, it contains all $(3\alpha, c)$-local cuts $(U, P_*\cut U)$ in $G_*[P_*]$ such that $U$ induces a connected subgraph. 
By Lemma \ref{lem:special-terminal-simple-cut},
the minimum $(3\alpha, c)$-local cut in $G_*$ must have a side of volume at most $3\alpha$ that is also connected. Then by Lemma \ref{lem:special-terminal}, $(U_2\cap P, P\cut U_2\cap P)$ is a minimum local cut in $G[P]$. Therefore, $(U_2\cap V, V\cut U_2\cap V)$ is a minimum local cut in $G[\mathcal P]$.
\end{proof}


In summary, a one-level global minimum cut sparsifier data structure is a tuple $(G, \mathcal P, CC, F, \Lambda_{G_*})$, where $(G, \mathcal P, CC, F)$ is a one-level terminal edge connectivity sparsifier, and $\Lambda_{G_*}$ is a priority queue of vertex sets.

\vspace{.15cm}\noindent \textbf{Multi-level Sparsifier Definition \ \ }
A multi-level global minimum cut data structure 
    $$\set{(G^{(j)}, \mathcal P^{(j)}, CC^{(j)}, F^{(j)},\Lambda_{G_*^{(j)}}) }_{j=1}^{\ell}$$
 contains a set of $\ell$ one-level  global minimum cut data structures such that conditions (c1), (c2) and (c3) in Section~\ref{sec:multi-s-t-ds} hold.

For a multi-level global minimum cut data structure, if we collect all one-level terminal edge connectivity data structures inside each one-level global minimum cut data structure for all the levels, 
these one-level terminal edge connectivity data structures together form a multi-level terminal edge connectivity data structure.
Thus, the multi-level minimum cut data structure can be viewed as a multi-level terminal edge connectivity data structure with additional priority queues on each level. 

Throughout this section, we assume that the parameters $c, \phi, c', \gamma$, and $ \alpha$ defining the one-level and multi-level global minimum cut data structure are set such that $c'\geq c$ always holds, $\gamma$ is always greater than $c'$, and $\alpha\geq \ceil{c/\phi}$.

\section{Initialization and Update Algorithms for the Global Minimum Cut Data Structure}\label{sec:algorithms}
With the new global minimum cut data structure defined, we proceed with presenting the algorithms to initialize and update it. In Sections \ref{sec:one-level-algs} and \ref{sec:multi-level-algs}, we will discuss the algorithms for one-level data structure and multi-level data structure, respectively. In Section \ref{sec:main-proof}, we will use these algorithms to prove Theorem \ref{thm:main}.
\subsection{One-level Initialization and Update Algorithms}\label{sec:one-level-algs}

\subsubsection{Subroutine: Cut Enumeration}
We begin with a subroutine to find all cuts in which a given vertex belongs to a side with bounded volume.
In Section~\ref{sec:one-level-init}, we will use this subroutine to find all the local cuts.

We define a subroutine called \texttt{Enumerate-Cuts}, which is invoked in the one-level global minimum cut data structure initialization and update algorithms. 
\texttt{Enumerate-Cuts} takes arguments $G, \mathcal P, G_*, v, c, \alpha$, where $\mathcal P$ is a partition of $G$, $G_* = (V_*, E_*)$ is the auxiliary graph of $G$ w.r.t. $\mathcal P$ (Definition \ref{def:aux-graph}), 
$v$ is a vertex in $G_*$, $c$ and $\alpha$ are positive integers. The returned value is a list of subsets of $V_*$ that can represent a cut in $G_*$. We will later use these subsets to construct $\Lambda_{G_*}$. 

\SetKwFunction{FLocalEnum}{Enumerate-Cuts}
\SetKwFunction{FLocal}{Minimum-Local-Cut}
\SetKwFunction{FCutRecursive}{Cut-Enumeration-Recursive}
\SetKwFunction{FCutVerify}{Atomic-Cut-Verification}

Algorithm \ref{alg:enum-new} is the detailed implementation of \texttt{Enumerate-Cuts}.

\begin{algorithm}[htb]
\caption{Cut Enumeration}\label{alg:enum-new}
\KwIn{$G_*=(V_*, E_*)$\newline 
$v$: a vertex in $G_*$, \newline $ \alpha, c$: positive integers
}
\KwOut{$\mathcal{C}$: a list of subsets of $V_*$} 
\SetKwFunction{FLocalEnum}{Enumerate-Cuts}
\SetKwFunction{FLocal}{Minimum-Local-Cut}
\SetKwFunction{FCutRecursive}{Cut-Enumeration-Recursive}
\SetKwFunction{FCutVerify}{Atomic-Cut-Verification}
\SetKwProg{Fn}{Function}{:}{}
\Fn{\FLocalEnum{$G_*, v, \alpha, c$}} {
    $T_{\mathcal P}\gets$ the set of special terminals added when constructing $G_*$, as per Definition \ref{def:aux-graph}\;
    $\mathcal{C} \gets \emptyset$\;
    \FCutRecursive($G_*, v, \alpha, c, \emptyset, \mathcal{C}$)\;
    \Return $\mathcal{C}$
}
\Fn{\FCutRecursive{$G_*, v, \alpha, c, F,  \mathcal{C}, T_{\mathcal P}$}} {
    \lIf{$c < 0$} {\Return}
    Run a DFS from $v$ in $G_*$ until the DFS stops or $\vol_{G}(U) > 3\alpha$, where $U$ is the set of visited vertices\label{alg:new-enum:line7}\; 
    \If{
    $F$ is exactly the set of edges between $U$ and $V\cut U$ in $G_*$, $\sum_{(u, w, \beta)\in F}\beta \leq c$, and $\vol_{G}(U)+ \sum_{(u, w, \beta)\in F}\beta\leq 3\alpha$} {\label{alg:new-enum:line8}
        $\mathcal{C} \gets \mathcal{C} \cup \{U\}$\;
        \Return
    }

    \For{every edge $(u, w, \beta)$ of $G_*$ such that $u, w\in U$ and $u,v\not\in T_{\mathcal P}$}{\label{alg:new-enum:line12}
        $F \gets F\cup \set{(u, w, \beta)}$\;
        Remove edge $(u, w, \beta)$ from $G[\mathcal P]$

            \FCutRecursive{$G_*,  v, \alpha,c-\beta,  F,  \mathcal{C}$}\;
        Add edge $(u, w, \beta)$ back to $G[\mathcal P]$\;
    }
}
\end{algorithm}
\begin{lemma}\label{lem:enum-new}
Let $G = (V, E)$ be a graph, $\mathcal P$ be a partition of $G$, $G_*$ be the auxiliary graph of $G$ w.r.t. $\mathcal P$.
For a given vertex $v$ in $G$ and two integers $\alpha, c > 0$,
Algorithm \ref{alg:enum-new} \FLocalEnum
finds all subsets of $V_*$, denoted by $U$, such that 
\begin{itemize}
    \item There exists a $P\in\mathcal P$ such that the corresponding $P_*$ in $G_*$ (Definition \ref{def:aux-graph}) contains $U$
    \item $v\in U$
    \item  $U$ induces a connected subgraph in  $G_*$
    \item $(U, P_*\cut U)$ is a $(3\alpha, c)$-local cut in $G_*[P_*]$ with respect to terminal set $\endpoints(\partial_G(\mathcal P))\cap P \cup\set{t_P}$
\end{itemize}

The running time of the algorithm is $O(\alpha)^{c}$.
\end{lemma}
\begin{proof}
We check the requirements one by one.

\begin{itemize}
    \item Since the DFS is conducted on $G_*$ and $U$ is the set of visited vertices, by the definition of $P_*$  (Definition \ref{def:aux-graph}), there exists a $P_*$ such that $U\subset P_*$.
    \item Since the DFS (line \ref{alg:new-enum:line7}) starts from $v$, each $U$ contains $v$.
    \item Since each $U$ is obtained from a DFS, $G_*[U]$ is connected.
    \item \begin{itemize}
        \item Since the DFS stops when the visited volume reaches $3\alpha$, the volume of $U$ in $G_*$ is at most $3\alpha$.
        \item Since the number of edges between $U$ and $V_*\cut U$ is at most $c$ (line \ref{alg:new-enum:line8}), the size of $(U, V_*\cut U)$ is at most $c$ in the returned pairs.
        \item Fix a vertex $t\in \endpoints(\partial_G(\mathcal P))\cap P$. $(t, t_P)$ is not in  the cut-set of $(U, V_*\cut U)$ (line \ref{alg:new-enum:line12}), $t$ and $t_P$ are on the same side of cut $(U, V_*\cut U)$. Therefore, all terminals in $ \endpoints(\partial_G(\mathcal P))\cap P\cup \set{t_P}$ are on the same side of the cut.
    \end{itemize}
\end{itemize}




To show that all $U$ satisfying the constraints are added to $\mathcal C$, let $U'$ be a set of vertices satisfying the desired conditions, $E'$ be the edge set between $U'$ and $V_*\cut U'$. We look at the recursion tree of \texttt{Cut-Enumeration-Recursive}. Since $U$ induces a connected subgraph in $G_*$, no proper subset of $E'$  form another cut in $G_*$.
Since $v\in U'$ and the volume of $U'$ in $G_*$ is at most $3\alpha$, by performing a DFS from $v$ until the visited volume is strictly greater than $3\alpha$, it is guaranteed that there exists an edge $e\in E'$ such that both endpoints of $e$ are visited. 

Suppose after a depth $i$ recursion, $i$ edges of $E'$ are added to $F$. $E'\cut F$ is the cut-set of cut $(U',V_*\cut U')$ in graph $G_*\cut F$. The algorithm performs another DFS from $v$ that stops when the volume of visited vertices is greater than $3\alpha$. Since $\vol_{G_*\cut F}(U')\leq \vol_{G_*}(U')\leq 3\alpha$, it is still guaranteed that there is at least one edge $e\in E'\cut F$ such that both endpoints of $e$ are visited. By induction, at the end of this chain of recursion, $F=\partial_{G_*}(U')$. Then, after one more recursion, $U'$ is added to $\mathcal C$.

Now we prove the time complexity. In each recursion, the DFS takes $O(\alpha)$ time. The conditions on line \ref{alg:new-enum:line8} can be checked by running another DFS and takes another $O(\alpha)$ time. Since the volume of the visited vertices is at most $\alpha$, at most $O(\alpha)$ edges are enumerated in line \ref{alg:new-enum:line12}. Therefore, the recursion tree is an $\alpha$-ary tree of depth at most $c$, and can be completed in time $O(\alpha)^{c}$. 
\end{proof}
\subsubsection{One-level Initialization}\label{sec:one-level-init}
We initialize the global minimum cut data structure by separately initializing the underlying terminal edge connectivity sparsifier and $\Lambda_{G_*}$.
\SetKwFunction{Fdeltainit}{Min-Cut-One-Level-Initialization}

\begin{algorithm}[htb]
\SetKw{Continue}{continue}
\caption{One-level Minimum Cut Data Structure Initialization}\label{alg:one-level-init-new}
\KwIn{$G$: a graph,\newline 
$c, \phi, c', \gamma, \alpha$: parameters.
}

\SetKwProg{Fn}{Function}{:}{}

\Fn{\Fdeltainit}
{
$(G,\mathcal P, CC, F)\gets \FTerminalInitialization(G, c', \phi, \gamma)$\;
\For{all $P\in\mathcal P$} {\label{alg:one-init-new-line4}
    Add vertex $t_P$ to $G$\;
    \For {each $t\in \endpoints(\partial_G(\mathcal P))\cap P$} {
        Add edge $(t_P, t)$ to $G$ \label{alg:one-init-new-line7}\; 
    }
}
$G_* = (V_*, E_*)\gets$ the resulting graph\;
    \For{all $v\in V_*$}{
        $\mathcal C\gets\texttt{Enumerate-Cuts}(G_*, v, \alpha, c)$\label{alg:one-init-new-line6}\;
        \For{$U\in \mathcal C$}{
                Insert $U$ into $\Lambda_{G_*}$ with priority $\abs{\partial_{G_*}(U)}$\hfill\tcp{$\abs{\partial_{G_*}(U)}$ can be computed and returned by \texttt{Enumerate-Cuts}. We do not need to find all edges in the cut-set here}
        }
        
    }
\Return $(G,\mathcal P, CC, F,  \Lambda_{G_*})$\;
}
\end{algorithm}

\begin{lemma}\label{lem:one-init-new}
For a graph $G$ with parameters $c,\phi, \gamma,\alpha$ satisfying $c'\geq c$, $\gamma\geq c'$, and $\alpha\geq \ceil{c'/\phi}$, Algorithm \ref{alg:one-level-init-new} computes a one-level global minimum cut data structure of $G$ w.r.t. parameters $c$, $\phi$, $c'$, $\gamma$, and $\alpha$ in time $m^{1+o(1)}O(\alpha)^{O(c')}$.
\end{lemma}
\begin{proof} By Lemma \ref{lem:old_main_initialization}, \texttt{Edge-Connectivity-One-Level-Initialization} initializes the edge connectivity sparsifier correctly. We only need to show that after the execution of Algorithm \ref{alg:one-level-init-new},  $\Lambda_{G_*}$ is correctly initialized.

Recall the definition of $\Lambda_{G_*}$. It should contain exactly all subsets $U$ such that  
\begin{itemize}
    
    \item There exists a $P\in\mathcal P$ such that $U\subset P_*$, where $P_* = P\cup \set{t_P}$ 
    \item  $(U, P_*\cut U)$ is an $(3\alpha, c)$-local cut in $G_*[P_*]$  for some $\alpha\geq \ceil{c'/\phi}$. The terminal set is $\endpoints(\partial_G(\mathcal P))\cap P \cup\set{t_P}$.
\end{itemize}
Fix such a $U$ with corresponding $P_*$. Let $v$ be an arbitrary vertex in $U$. By Lemma \ref{lem:enum-new}, $U$ is in the list of subsets returned by \texttt{Enumerate-Cuts} (line \ref{alg:one-init-new-line6}).

For any $U$ that is not supposed to be in $\Lambda_{G_*}$. Lemma \ref{lem:enum-new} guarantees that $U$ would not be returned by \texttt{Enumerate-Cuts}.

For the time complexity, by Lemma \ref{lem:old_main_initialization}, \texttt{Edge-Connectivity-One-Level-Initialization} takes $m^{1+o(1)}c'^{O(c')}/\text{poly}(\phi)$ time. Constructing $G_*$ takes at most $O(m)$ time (line \ref{alg:one-init-new-line4} to \ref{alg:one-init-new-line7}). Each invocation of \texttt{Enumerate-Cuts} runs in $O(\alpha)^{c}$ time, and returns $O(\alpha)^{c}$ vertex subsets. Therefore, the total time needed to insert all subsets into $\Lambda_{G_*}$ is at most $n\cdot O(\alpha)^{c}\cdot O(c\log \alpha + \log n)$. 

Therefore, the overall time complexity is
\[m^{1+o(1)}c'^{O(c')}/\text{poly}(\phi) + n\cdot O(\alpha)^{c}\cdot O(c\log \alpha + \log n) = m^{1+o(1)} O(\alpha)^{O(c')}\]
\end{proof}
\subsubsection{One-level Update}\label{sec:one-level-update}
When updates occur, in addition to updating  the terminal cut sparsifier, we need to also update $\Lambda_{G_*}$. To achieve this, we remove and add some vertex sets  and make sure that after the update, $\Lambda_{G_*}$  still satisfies Definition \ref{def:lambda} for the updated graph.


To achieve this, we select a list of vertices $S$,  remove from $\Lambda_{G_*}$ all $U$  that contain at least one vertex in $S$. Then we run \texttt{Enumerate-Cuts}  from every vertex in $S$ and insert the returned vertex sets back to  $\Lambda_{G_*}$. We observe that it suffices to let $S$ contain only the following vertices:
\begin{enumerate}
    \item Terminals vertices that switch to non-terminals after the updates, or non-terminals vertices that switch to terminals after the update. These are also the vertices returned by the terminal cut sparsifier update algorithm \texttt{Edge-Connectivity-One-Level-Update} (Lemma \ref{lem:old_main}).
    \item Vertices that are endpoints of edges involved in the update sequence.
\end{enumerate}

We formally state and prove the observation above with the following lemma.

\begin{lemma}\label{lem:delta-update-new} Let $G$ be a graph, and $\Lambda_{G_*}$ be the priority queue in the one-level global minimum cut data structure of $G$ with parameters $c, \phi,c', \gamma$ and $\alpha$.
After an update sequence updates $G$ to $G'$, suppose an expander decomposition $\mathcal P$ of $G$ is updated to $\mathcal P'$ of $G'$ such that for each $P'\in \calP'$, $P'$ either contains only a single vertex or is a subset of some $P \in \calP$ such that $G[P'] = G'[P']$. 
Let $S$ be a list of vertices $v$ that satisfy one of the three following conditions:
\begin{itemize}
    \item $v$ is a terminal vertex before the update and a non-terminal after the update
    \item $v$ is a non-terminal vertex before the update, and a terminal after the update 
    \item $v$ is an endpoint of an edge involved in the update sequence. 
\end{itemize}
Let $\Lambda'_{G_*} = \set{U\in \Lambda_{G_*}: U\cap S\neq\emptyset}$, let $\Lambda''_{G_*}$ be the set of all vertex set $U'$ such that 
\begin{enumerate}[label=(\alph*)]
    \item $U'\cap S\neq \emptyset$
    \item There exists $P'\in \mathcal P'$, such that $U'\subset P'_*$, where $P'_*$ is the cluster corresponding to $P'$ in $\mathcal P'_*$, per Definition \ref{def:aux-graph}.
    \item $U'$ induces a connected subgraph of $G'_*$

    \item $(U', P_*'\cut U')$ is a $(3\alpha, c)$-local cut in $G'_*[P'_*]$ with respect to terminal set  $\endpoints(\partial_{G'}(\mathcal P'))\cup \set{t_{P'}}$
\end{enumerate}
 $\Lambda_{G'_*} = \Lambda_{G_*}\cut \Lambda'_{G_*}\cup \Lambda''_{G_*}$, ordered by the the size of each  $(U', P_*'\cut U')$, is a valid priority queue  in the one-level global minimum cut data structure of $G'$ w.r.t parameters $c, \phi,c', \gamma$ and $\alpha$.
\end{lemma}

\begin{proof}

After comparing Definition \ref{def:lambda} to conditions (b), (c) and (d),  we essentially need to prove 
\begin{enumerate}
    \item For any $U' $ that should be in $\Lambda_{G'_*}$, if $U'\not\in\Lambda''_{G_*}$, then $U'\in\Lambda_{G_*}$  and $U'\not \in \Lambda'_{G_*}$.
    
    Fix  a $U'\in\Lambda_{G'_*}\cut \Lambda''_{G_*}$. It has to be true that $U'\cap S = \emptyset$.  In this case, there exists a $P'\in\mathcal P'$ and a $P\in\mathcal P$ such that $U'\subset P'_*\subset P_*$. $U'$ has the same volume in $G_*[P_*]$ and $G'_*[P'_*]$, and thus $U'\in \Lambda_{G_*}$. Since $U'\cap S = \emptyset$, $U'\not\in \Lambda'_{G_*}$.
    
    \item For  any $U''\in \Lambda_{G_*}$ that does not satisfy all of conditions (a) (b) (c) and (d), $U''\in \Lambda'_{G_*}$.
    
    Fix a $U''\in \Lambda_{G_*}$ that should not be in $\Lambda_{G'_*}$. We only need to show that $U''\cap S\neq\emptyset$. Assume to the contrary. There exists a $P'\in\mathcal P'$ and a $P\in\mathcal P$ such that $U''\subset P'_*\subset P_*$. Therefore, $U''$ induces the same subgraph in $G_*$ and $G'_*$ and should also be in a pair in $\Lambda_{G'_*}$. We reached a contradiction.
\end{enumerate}
This concludes the proof.
\end{proof}

\begin{Remark}
Since the update of $\Lambda_{G_*}$ does not involve the terminal cut sparsifier, $c'$ and $\gamma$ can be any number satisfying $c'\geq c$ and $\gamma\geq c'$. $c'$ and $\gamma$ do not have to be the same for the data structures before and after the update.
\end{Remark}

We have an update routine in Algorithm \ref{alg:one-up-new}:  

\begin{algorithm}[htb!]
\caption{Update One Level Minimum Cut Data Structure}\label{alg:one-up-new}
\KwIn{
$(G, \mathcal P, CC, F, \Lambda_{G_*})$: one-level minimum cut data structure, w.r.t. parameters $c, \phi, c', \gamma, \alpha$ \newline
$\updateseq$: update sequence for multigraph $G$
}
\KwOut{
$(G', \mathcal P', CC', F',\Lambda_{G'_*})$: one-level minimum cut data structure for the up-to-date graph,  w.r.t parameters $c, \phi/\factor, c, \gamma, \alpha$ \newline
$\updateseq'$: update sequence that should be applied to the next sparsifier}
\SetKwFunction{Fenuupdate}{Update-Enumeration}
\SetKwFunction{Fdeltaupdate}{Minimum-Cut-One-Level-Update}
\SetKwProg{Fn}{Function}{:}{}
\SetKw{Continue}{continue}

\Fn{\Fdeltaupdate}
{
    $(G', \calP', CC', F'), \updateseq', S\gets $ \FTerminalUpdatae{$(G, \calP, CC, T), \updateseq, c', \phi, \gamma$}\;
    Add all vertices involved in $\updateseq$ to $S$\;\label{alg:one-up-new:line5}
    \tcp{Update $G_*$}

    \For{all $P'\in\mathcal P'$ such that $P'\cap S\neq \emptyset$} {\label{alg:one-up-new:line4}
        $P\gets$ the cluster in $P$ such that $P'\subset P$, as per Lemma \ref{lem:old_main}\;
        Remove $t_P$ and all incident edges from $G_*$\;
        Add vertex $t_{P'}$ to $G_*$\;
        \For {all $t\in \endpoints(\partial_{G'} (\mathcal P'))\cap P'$}{
            Add edge $(t, t_{P'})$ to $G_*$\label{alg:one-up-new:line9}\;
        }
    }
    $G'_* = (V'_*, E'_*)\gets$ the resulting graph\label{alg:one-up-new:line10}\;
    \tcp{Update $\Lambda_{G_*}$}
    \For {all $v\in S$}{\label{alg:one-up-new:line15}
        $P'\gets$ the cluster in $\mathcal P'$ that contains $v$\;
        Remove from each $\Lambda_{G_*}$ all $U$ such that $v\in U$\label{alg:one-up-new:line18}\;
        $\mathcal C\gets $\texttt{Enumerate-Cuts}$(G'_*, v, \alpha, c )$\label{alg:one-up-new:line18-2}\;
        \For{each $U\in\mathcal C$} {\label{alg:one-up-new:line21}
            $j\gets $ the size of $(U, V\cut U)$\;
            Insert $U$ into $\Lambda_{G_*}$ with priority $j$\label{alg:one-up-new:line26}
        }
    }
    Change the name of  $\Lambda_{G_*}$ to $\Lambda_{G_*'}$

    \Return $(G', \mathcal P', CC', F',\Lambda_{G_*'})$, $\updateseq'$
}
\end{algorithm}

\SetKwFunction{Fenuupdate}{Update-Enumeration}
\SetKwFunction{Fdeltaupdate}{Minimum-Cut-One-Level-Update}

\begin{lemma}\label{lem:one-up-new}
Given a one-level minimum cut data structure $(G, \calP, CC, F, \Lambda_{G_*})$ with parameters $c, \phi, c',  \gamma, \alpha$ satisfying the following inequalities 
\[c' = c^2 + 2c,\gamma > c',  \alpha \geq (c^2 + 2c)\cdot 2^{O(\log^{1/3}n\log\log n)} / \phi, \text{ and } \phi \leq \frac{1 }{2^{\log^{3/4} n}},\]
and an update sequence $\updateseq$ of $G$ which consists of vertex/edge insertions/deletions, 
there is an algorithm 
\Fdeltaupdate of running time $O(|\updateseq|)\cdot O(\alpha)^{O(c^2)}$ which updates the data structure to $(G', \calP', CC', F',\Lambda_{G'_*})$ with parameters $c, \phi/2^{O(\log^{1/3}n\log\log n)}$, $c$, $\gamma, \alpha$, and outputs an update sequence $\updateseq'$ satisfying the following properties:
\begin{enumerate}
 
    \item $G'$ is the resulted graph of applying $\updateseq$ to $G$.
    \item $\calP'$ is a $\phi/2^{O(\log^{1/3}n\log\log n)}$-expander decomposition of $G'$. Furthermore, each cluster $P'\in\mathcal P'$ either contains only a single vertex or is a subset of a cluster $P\in \mathcal P$ such that $G'[P'] = G[P']$.
    \item $CC'$ is a $c$-cut containment set for $G'$ and $\calP'$. 
    \item $F'$ is an unweighted spanning forest of $G'\setminus CC'$.     
    
    \item The sparsifier maintained in the data structure is $\sparsifier_{c, \gamma}(G', \calP', CC', F')$.
    \item $\updateseq'$ updates $\sparsifier_{c', \gamma}(G, \calP, CC, F)$ to $\sparsifier_{c, \gamma}(G', \calP', CC', F')$. The length of $\updateseq'$ is at most $\abs{\updateseq}(10c)^{O(c)}$
    \item $\Lambda_{G'_*}$ is a priority queue containing every vertex set $U'$ such that
    \begin{itemize}
 \item There exists a $P'\in\mathcal P'$ such that $U'\subset P'_*$, where $P'_*$ is the corresponding cluster of $P'$ in $G'_* = (V_*', E_*')$
    \item $U'$ induces a connected subgraph in $G'_*$
    \item  $(U', P'_*\cut U')$ is an $(3\alpha, c)$-local cut in $G'_*[P'_*]$  for some $\alpha\geq \ceil{c/\phi}$. The terminal set is $\endpoints(\partial_{G'}(\mathcal P'))\cap P' \cup\set{t_{P'}}$.
\end{itemize}
The priority of vertex subsets in $\Lambda_{G'_*}$ is given by the size of $(U', V'_*\cut U')$ in increasing order.

\end{enumerate}
\end{lemma}
\begin{proof}

Properties 1, 2, 3, 4, 5, 6 are proven in Lemma \ref{lem:old_main}. For a $P'\in\mathcal P'$, by 2, there is a $P\in\mathcal P$ such that $P'\subset P$. If $P'$ does not contain any terminal vertex that is a non-terminal before the update, or any non-terminal vertex that is a terminal before the update, or any vertex that is involved in the update sequence, then  $P'=P$ and $P'$ and $P$ have the same set of terminals. The $c$-cuts in $P'$ are the same as those in $P$. Therefore, in order to update $G_*$ and $\Lambda_{G_*}$, we only need to consider clusters $P'\in\mathcal P'$ that contains at least one vertex in $S$ (line \ref{alg:one-up-new:line4}). For these clusters, we remove the old special terminal and add new special terminals for each of them (line \ref{alg:one-up-new:line4} to \ref{alg:one-up-new:line9}). Since these are all of the clusters affected by the update, in line \ref{alg:one-up-new:line10}, $G'_*$ is same as the graph obtained by directly adding special terminals and edges to the up-to-date graph $G'[\mathcal P']$.

Since $\Lambda_{G'_*}$ is obtained by removing from $\Lambda_{G_*}$ all $U$ such that $U$ contains at least one vertex in $S$ and then adding all new $U'$ returned by running \texttt{Enumerate-Cuts} from vertices in $S$, by Lemma \ref{lem:delta-update-new}, the resulting data structures satisfy Definition \ref{def:lambda}.

For the time complexity, Lemma \ref{lem:old_main} shows that \texttt{Edge-Connectivity-One-Level-Update} runs in $O\wrap{\abs{\updateseq}(c/\phi)^{O(c^2)}}$ time (because $c'$ is $O(c^2)$), and that the vertex set $S$ returned by this subroutine contains at most $O(\abs{\updateseq})$ vertices . After adding vertices involved in $\updateseq$ to $S$ (line \ref{alg:one-up-new:line5}), $S$ still has $O(\abs{\updateseq})$ vertices. Therefore,  line \ref{alg:one-up-new:line9} is executed for at most $O(\abs{\updateseq})$ times. The overall time for updating $G_*$ is at most $O(\abs{\updateseq})$.

The block from line \ref{alg:one-up-new:line15} to line \ref{alg:one-up-new:line26}  is executed for at most $O(\abs{S})  =O(\abs{\updateseq})$ times.
Since for a fixed $v$, every $U$ that contains $v$ can be enumerated by running \texttt{Enumerate-Cuts} again on a copy of the data structure before the update, line \ref{alg:one-up-new:line18} takes $O(\alpha)^{c}\cdot O(c\log\alpha)$ time. Since $\alpha> c/\phi$ and $0< \phi< 1$, $O(\alpha)^{c}\cdot O(c\log\alpha) = O(\alpha)^{O(c)}$.
By Lemma \ref{lem:enum-new}, the invocation of \texttt{Enumerate-Cuts} on line \ref{alg:one-up-new:line18-2} takes $O(\alpha)^{O(c)}$ time. By an argument similar to Lemma \ref{lem:one-init-new}, the block from line \ref{alg:one-up-new:line21} to \ref{alg:one-up-new:line26} takes $O(\alpha)^{O(c)}$ time. Therefore, the block updating  $\Lambda_{G_*}$ takes $O(\abs{\updateseq} O(\alpha)^{O(c)})$ time. Since $\alpha> c/\phi$, the overall time complexity of Algorithm \ref{alg:one-up-new} is $O(\abs{\updateseq} O(\alpha)^{O(c^2)})$.
\end{proof}


\subsection{Multi-level Initialization and Update Algorithms}\label{sec:multi-level-algs}
With one-level initialization and update algorithms, we can define the multi-level initialization and update algorithms.
\SetKwFunction{Fmultiinit}{Multi-Level-Min-Cut-Initialization}
\SetKwFunction{Fmultiupdate}{Multi-level-Update}
\begin{lemma}\label{lem:multi-init-new}
For a graph $G$ with $n$ vertices and $m$ edges and parameters $c, \phi, c', \gamma, \alpha$ satisfying the following inequalities 
\[c' > c, \gamma > c', \alpha \geq c / \phi, \text{ and } \phi \leq \frac{1 }{2^{\log^{3/4} n}}\] 
Algorithm \ref{alg:multi-init-new} constructs an at most $O(\log^{1/10} n)$-level global minimum cut data structure for $G$ with parameters $c, \phi, c', \gamma, \alpha$
in time \[m^{1+o(1)} \cdot O(\alpha)^{O(c')}\]
\end{lemma}

\begin{algorithm}[htb]
\caption{Multi-Level Minimum Cut Data Structure Initialization} \label{alg:multi-init-new}
\KwIn{$G$: graph before update,\newline
$c, \phi,c', \gamma, \alpha$: parameters.
}
\KwOut{$\{(G^{(i)}, \calP^{(i)}, CC^{(i)}, F^{(i)},    \Lambda_{G^{(i)}_*} \}_{i=1}^\ell$: multi-level minimum cut data structure
}

\SetKwFunction{Foneinit}{$c$-terminal-one-level-init}
\SetKwProg{Fn}{Function}{:}{}
\Fn{\Fmultiinit{$G, c, \phi,c', \alpha, \gamma$}}{
$\ell\gets 0, G^{(0)}\gets G$\;
\Repeat{$\calP^{(\ell - 1)}$ contains only one cluster}{
$\wrap{G^{(\ell)}, \calP^{(\ell)}, CC^{(\ell)}, F^{(\ell)},  \Lambda_{G_*^{(i)}}}\gets $ \Fdeltainit{$G^{(\ell)}, c, \phi, c',\gamma, \alpha$}\;
$G^{(\ell + 1)}\gets \sparsifier_{c', \gamma}(G^{(\ell)}, \calP^{(\ell)}, CC^{(\ell)}, F^{(\ell)})$\;
    $\ell\gets \ell + 1$\;}
    \Return $\set{(G^{(i)}, \mathcal P^{(i)},  CC^{(i)}, F^{(i)}, \Lambda_{G_*^{(i)}})}_{i=0}^{\ell-1}$
}

\end{algorithm}

\begin{proof}
By the definition of the multi-level global minimum cut data structure and the algorithm, the output is a multi-level global minimum cut data structure for the input graph with parameters $c, \phi, \gamma$, and $ \alpha$. 
By Theorem~\ref{thm:js21main}, 
$\ell $ is at most $ O(\log^{1/10} n)$.

By Lemma~\ref{lem:one-init-new}, the initialization of a one-level minimum cut data structure takes \[m^{1+o(1)} \cdot O(\alpha)^{O(c')}\] time for a graph with $m$ edges. 
Since the multi-level data structure has at most $O(\log^{1/10}n)$ levels, \\\Fdeltainit is called at most $O(\log^{1/10}n)$ times, the total time of \\executing \Fdeltainit is $O(m ^{1+o(1)}\log n) \cdot O(\alpha)^{O(c')}$.
Therefore, the total running time of the algorithm is \[O(m ^{1+o(1)}) \cdot O(\alpha)^{O(c')}\]
\end{proof}

In our multi-level update algorithm, for a given update sequence, we use the one level update algorithm to update the minimum cut data structure for each level starting from level 0. After the update on  level $i$, the algorithm returns a new update sequence for the level $i+1$.
If at some level, the update sequence contains too many updates such that the resulted sparsifier is too large, we reconstruct the sparsifier to make sure that all the sparsifiers get smaller as the level increases. 

\begin{algorithm}[htb]
\caption{Multi-Level Minimum Cut Data Structure Update}\label{alg:multi-up-new}
\KwIn{$\set{G^{(i)}, \mathcal P^{(i)}, CC^{(i)}, F^{(i)},   \Lambda_{G_*^{(i)}} }_{i=0}^\ell$: multi-level global cut data structure,\newline
$\updateseq$: a sequence of update to  graph $G^{(0)}$,\newline
$c, \phi, c'=c^2 +2c,  \gamma, \alpha$:  parameters of the multi-level global cut data structure
}
\KwOut{$\set{G'^{(i)}, \mathcal P'^{(i)}, CC'^{(i)}, F'^{(i)},   \Lambda_{G_*'^{(i)}}}_{i=0}^{\ell'}$: multi-level global cut data structure for the graph after applying $\updateseq$\newline}
\SetKwFunction{Foneupdate}{$c$-terminal-one-level-update}
\SetKwProg{Fn}{Function}{:}{}
\Fn{\Fmultiupdate{}}
{
$i\gets 0, \updateseq^{(0)}\gets \updateseq$\;
\While{$\updateseq^{(i)}$ contains no more than $n^{(i)} / \log n$ updates, where $n^{(i)}$ is the number of vertices in $G^{(i)}$}{
    $( G'^{(i)}, \mathcal P'^{(i)}, CC'^{(i)}, F'^{(i)},   \Lambda_{G_*'^{(i)}}  ), \updateseq^{(i+1)} \gets$ \Fdeltaupdate
    ($(G^{(i)}, \mathcal P^{(i)}, CC^{(i)}, F^{(i)},  \Lambda_{G_*^{(i)}}   ), \updateseq^{(i)}, c, \phi,c^2 + 2c, \gamma, \alpha)$\; 
    $i \gets i+1$\;
}
$\phi\gets \phi/2^{O(\cdot \log^{1/3}n\log^{2/3}\log n)}$\;

$\set{G'^{(j)}, \mathcal P'^{(j)}, CC'^{(j)}, F'^{(j)},  \Lambda_{G_*'^{(i)}}}_{j = i}^{\ell'}\gets $ \Fmultiinit{$G'_{(i)}, c, \phi, c, \gamma, \alpha$}\label{alg:multi-up-new-line7}\;
\Return $\set{G'^{(i)}, \mathcal P'^{(i)}, CC'^{(i)}, F'^{(i)}, \Lambda_{G_*'^{(i)}}}_{i = 0}^{\ell'}$\;
}
\end{algorithm}

\begin{lemma}\label{lem:multi-up-new}
Given a multi-level minimum cut data structure $\set{G^{(i)}, \mathcal P^{(i)}, CC^{(i)}, F^{(i)}, \Lambda_{G_*^{(i)}}}_{i=0}^\ell$ with parameters $c, \phi, c', \alpha, \gamma$ satisfying the following inequalities 
\[c'= c^2 + 2c, \gamma > c' ,  \alpha \geq (c^2 + 2c)\cdot 2^{O(\log^{1/3}n\log\log n)} / \phi, \text{ and } \phi \leq \frac{1 }{2^{\log^{3/4} n}}\] 
such that the data structure is updated using at most  $O(\log \log n)$ update sequences since initialized,
and an update sequence $\updateseq$ of graph $G^{(0)}$, 
then 
Algorithm \ref{alg:multi-up-new} update the multi-level minimum cut data structure of $O(\log^{1/10} n\log \log n)$ levels with respect to the given update sequence with parameters $c, \phi / 2^{O(\log^{1/3}n\log^{2/3}\log n)}, \gamma, \alpha$ in time  \[O(\abs{\updateseq}O(\alpha )^{O(c^2 \log ^{1/10}n \log \log n)})\]
\end{lemma}

\begin{proof}
The correctness of the algorithm is implied by Lemma~\ref{lem:one-up-new} and Lemma~\ref{lem:multi-init-new}. 
Since the multi-level data structure has $O(\log^{1/10} n)$ levels after initialization (Lemma~\ref{lem:multi-init-new}), and after applying an update sequence, the number of levels increases by at most $O(\log^{1/10} n)$  (new layers are added in line \ref{alg:multi-up-new-line7} by running the initialization algorithm again), the total number of levels is  $O(\log^{1/10} n \log \log n)$ after applying $O(\log \log n)$ update sequences to the graph.

Now we bound the running time. By Lemma~\ref{lem:one-up-new}, at level $i$, the update sequence has length at most $|\updateseq|(10 c)^{O(i\cdot c)}$, and  the running time of \Fdeltaupdate for level $i$ is at most 
\[O(|\updateseq|) \cdot (10c)^{O(i\cdot c)}\cdot O(\alpha )^{O(c^2)} =O(|\updateseq|) \cdot  O(\alpha )^{O(i \cdot c^2)}\]
Together with the fact that the minimum cut data structure has at most $O(\log^{1/10} n \log \log n)$ levels, the overall running time is $O(\abs{\updateseq}O(\alpha )^{O(c^2 \log ^{1/10}n \log \log n)})$.
\end{proof}


\subsection{Returning the Global Minimum Cut and Proof of Theorem \ref{thm:main}}\label{sec:main-proof}
\SetKwFunction{Ffinal}{Dynamic-Min-Cut-DS}
\begin{lemma}\label{lem:update-new}
For any $c = (\log n)^{o(1)}$, 
there is a fully dynamic algorithm \Ffinal  which maintains a set of multi-level global minimum cut data structure such that each maintained multi-level global minimum cut data structure has $O(\log^{1/10} n \log \log n)$ levels all the time, and after processing each update, the algorithm provides the access to one of the maintained multi-level data structure   $\set{G^{(i)}, \mathcal P^{(i)}, CC^{(i)}, F^{(i)},\Lambda_{G_*^{(i)}}   }_{i=0}^\ell$ with parameters $c, \phi, c', \gamma$, and $\alpha$  for the up-to-date graph satisfying the following conditions:
\[\phi = 1/n^{o(1)},  c' \geq c, \gamma > c',  \text{ and } \alpha \geq c' / \phi\]
The initialization time of the algorithm is $m^{1+o(1)}$, and the update time of the algorithm is $n^{o(1)}$ per update.
\end{lemma}

\begin{proof}
Suppose parameter $\parametertimesub$ is defined as 
\begin{align*}
\parametertimesub \eqdef \left \lfloor \log \left( \frac{\log \log n / 100}{\log (4c)} \right)\right \rfloor.
\end{align*}
Since $c=(\log n)^{o(1)}$, we have $\parametertimesub= \omega(1)$ and $\parametertimesub= O(\log \log \log n)$.

We define parameters $c_i, \phi_i$ for all $0, \leq i\leq \parametertimesub$, and $\gamma, \alpha$ such that the parameters for a multi-level minimum cut data structure are $c, \phi_0, c_0, \gamma, \alpha$ after initialization, and 
after executing Algorithm \Fmultiupdate on the multi-level minimum cut data structure $i$ times, the  parameters for the multi-level minimum cut data structure are $c, \phi_i, c_i, \gamma, \alpha$.
We set $c_i$ as follows
\[c_{\parametertimesub} \eqdef c \text{ and } c_i \eqdef c_{i+1} (c_{i+1} + 2)  \text{ for } 0 \leq i < \parametertimesub.
\]
We set $\phi_{i}$ as follows
\[\phi_0 \eqdef \frac{1 }{2^{\log^{3/4} n}}, \phi_i \eqdef \frac{\phi_{i-1} }{ \truefactor} \text{ for } 1 \leq i \leq \parametertimesub, \text{where } \delta \text{ is a constant}. \]
We set $\gamma, \alpha$ as 
\[\gamma \eqdef c_{0} + 1, \alpha \eqdef \lceil c_0 / \phi_{\parametertimesub} \rceil\]
We have the following observations for our parameters
\begin{itemize}
    \item $c_i < (4c)^{2^{\parametertimesub}} = \log^{1/100} n$ for any $0 \leq i \leq \parametertimesub$.
    \item $\phi_i \leq  \frac{1}{2^{\log^{3/4} n}}$ and $\phi_i =  \frac{1}{2^{O(\log^{3/4} n)}}$ for any $0 \leq i \leq \parametertimesub$.
    \item $\gamma > c_i, \alpha \geq c_i / \phi_i$ for any $0 \leq i \leq \parametertimesub$.
    \end{itemize}
By Lemma~\ref{lem:multi-init-new} and Lemma~\ref{lem:multi-up-new}, if we have a multi-level minimum cut data structure initialized by Algorithm \Fmultiinit with parameters $c, \phi_0, c_0, \gamma, \alpha$, and the data structure has been updated $i$ times by Algorithm \Fmultiupdate for some $i\leq \parametertimesub-1$ such that after $j$-th for each $1 \leq j \leq i$, the parameters of the data structures are $c, \phi_j, c_j, \gamma, \alpha$, then if we use Algorithm \Fmultiupdate to update the data structure again, the parameters of the data structures are $c, \phi_{i+1},c_{i+1}, \gamma, \alpha$. By our choice or parameters and the time complexity of Lemma \ref{lem:multi-init-new} and \ref{lem:multi-up-new}, since $c_i < \log^{1/100} n$ and $\phi_i = 1/2^{O(\log^{3/4} n)}$, 
$O(\alpha)^{O(c^2\log ^{1/10} n \log \log n)}$ is always $n^{o(1)}$. Therefore, \Fmultiinit can initialize the data structure in $m^{1+o(1)}$ time and \Fmultiupdate can handle a sequence of updates with $n^{o(1)}$ time per update.

By our selection of parameters and Lemma~\ref{lem:fully_framework}, we obtained a fully dynamic algorithm with $m^{1+o(1)}$ initializaiton time and $n^{o(1)}$ update time.
\end{proof} 

\SetKwFunction{Fquery}{Fully-Dynamic-Min-Cut}
Now we prove Theorem~\ref{thm:main} using Algorithm \Fquery (Algorithm~\ref{alg:query-new}). In our dynamic algorithm, we maintain a set of multi-level minimum cut data structures such that after processing each update, the algorithm returns the access to one of the maintained multi-level minimum cut data structure with parameters $c, \phi, c', \alpha,  \gamma$ which is for the up-to-date graph.
Then for each one-level data structure $(G^{(i)}, \calP^{(i)},  CC^{(i)}, F^{(i)}, \Lambda_{G_*^{(i)}}  )$, one of the following cases hold: 
\begin{enumerate}
    \item The cut-set of any minimum cut of size $c$ for graph $G^{(i)}$ is also the cut-set of a minimum cut for $G$.

    \item The minimum cut of $G^{(i)}$ is of size greater than $c$.

\end{enumerate}

Furthermore, by Lemma \ref{lem:one-level-query}, the minimum cut of $G^{(i)}$ can be obtained by finding the smallest cut  from 
\begin{itemize}
    \item The cut induced by the cut-set of a minimum cut in $\sparsifier_{c,\gamma}(G^{(i)}, \mathcal P^{(i)}, CC^{(i)}, F^{(i)})$, which is also $G^{(i+1)}$
    \item A cut $(U\cap V, V_*\cut U\cap V)$, where $U $ is  the vertex set in $\Lambda_{G_*^{(i)}}$ with highest priority.
\end{itemize}


Recall that terminals in one level minimum cut data structures are boundary vertices of the expander decomposition. When a terminal cut exists in the one-level minimum cut data structure for level $i$,  there are more than one cluster in the expander decomposition in level $i$. Therefore, we can retrieve the minimum cut from sparsifiers of the current level, i.e., the one-level minimum cut data structure for level $i+1$. 
We have the following algorithm to return the minimum cut using the multi-level minimum cut data structure. 
\SetKwFunction{Fquery}{Fully-Dynamic-Min-Cut}
\SetKwProg{Fn}{Function}{:}{}
\begin{algorithm}[htb]
\caption{Fully Dynamic Minimum Cut}\label{alg:query-new}
\KwIn{$\updatesingle$: a multigraph vertex/edge insertion/deletion update
}

\Fn{\Fquery}{
$\set{G^{(i)}, \mathcal P^{(i)},  CC^{(i)}, F^{(i)},\Lambda_{G^{(i)}_*}}_{i=0}^\ell\gets \Ffinal(\updateseq)$\;
\For{$i\gets$ $\ell$ down to $0$}{
    $U\gets \text{highest priority item in }\Lambda_{G_*^{(i)}}$\label{alg:query-new-line7}\;
    $A_i\gets$ edges between $U$ and $V\cut U$ or $\perp$ if it doesn't exist\label{alg:query-new-line8}\;
    \If {$i < \ell$ and ($A_{i} = \bot$ or $|A_i| > |A_{i+1}|$)} {\label{alg:query-new-line10}
        $A_i \gets A_{i+1}$\;
    }
}
\If {$A_0 = \bot$} {\Return $\emptyset$} \Else{\Return $A_0$}
}
\end{algorithm}

\begin{proof}[Proof of Theorem~\ref{thm:main}]
We first prove the correctness of Algorithm~\ref{alg:query-new}. 
By Lemma~\ref{lem:update-new}, 
the multi-level minimum cut data structure obtained is for the up-to-date graph with $\ell = O(\log^{1/10} n\log \log n)$ levels.

We use induction to show that for any $0\leq i\leq \ell$, $A_i$ is the cut-set of a minimum cut of $G^{(i)}$ if the minimum cut of $G^{(i)}$ is of size at most $c$, otherwise, $A_i$ is $\bot$. 
Since $G^{(\ell)}$ only contains a single cluster, all the vertices of $G^{(\ell)}$ are non-terminal vertices. Therefore, if a minimum cut of size  $j\leq c$ exists, a side of this cut is stored in $\Lambda_{G_*^{\ell}}$.
Hence, $A_\ell$ is the cut-set of a minimum cut of  $G^{(\ell)}$ if the minimum cut of $G^{(\ell)}$ is of size at most $c$, otherwise, $A_\ell$ is $\bot$. 

For any $0\leq i < \ell$, by Lemma~\ref{lem:one-level-query}, if the minimum cut of $G^{(i)}$ is of size at most $c$, then it is either obtained from a minimum cut in the sparsifier ($A_{i+1}$ here), or a vertex set with the lowest priority in $\Lambda_{G_*^{(i)}}$ (line \ref{alg:query-new-line7} and \ref{alg:query-new-line8}).
In line \ref{alg:query-new-line10}, the algorithm sets $A_i$ to be the smaller of the two cuts above. Therefore, $A_i$ is the cut-set of a minimum cut of  $G^{(i)}$ if the minimum cut of $G^{(i)}$ is of size at most $c$, otherwise, $A_i$ is $\bot$.

To bound the running time, we note that for any $0 \leq i \leq \ell$, 
the minimum cut from $\Lambda_{G_*^{(i)}}$ can be obtained in $O(c\log n)$ time. Since the multi-level data structure has $O(\log^{1/10} n\log \log n)$ levels, the running time to find the cut-set of minimum $c$-cut is polylogarithmic. 
By Lemma~\ref{lem:update-new}, the processing time for each update is $n^{o(1)}$.
\end{proof}

\bibliography{reference}
\bibliographystyle{plain}
\appendix

\section{Dynamic Simple Graph to Dynamic Multigraph with A Constant Number of Neighbors}\label{sec:transform}
Let $\overline{G} = (\overline{V}, \overline{E})$ be the original dynamic simple graph with an arbitrary degree. 
For a fixed positive integer $c$,
we use the degree reduction technique~\cite{harary6graph} to transform $\overline{G}$ to a multigraph $G  = (V, E)$ 
such that every vertex has at most a constant number of distinct neighbors
as follows:
\begin{enumerate}
\item The vertex set $V$ of $G$ is \[V = \left\{v_{u, w} : (u, w) \in \overline{E}\right\} \cup \left\{v_{u, u} : u \in \overline{V}\right\}.\] 
\item 
For any edge $(u, w) \in \overline{E}$, add edge $(v_{u, w}, v_{w, u}, 1)$ to $E$.
\item For every vertex $u$ of $\overline{G}$ with degree at least $1$, 
let $w_{u, 0} = u$ and  $w_{u, 1}, w_{u, 2}, \dots, w_{u, \deg(u)}$ be the neighbors of $u$ in $\overline{G}$.
Add edge $(v_{u, w_{u, i}}, v_{u, w_{u, i+1}}, c+1)$ to $E$ for every $0 \leq i < \deg(u)$.
\end{enumerate}
To maintain the correspondence between $\bar G$ and $G$, for every $u \in \bar V$, we maintain the list of $w_{u, 0}, \dots, $ $w_{u, \deg(u)}$.

Suppose $\overline E' \subset \overline E$ is the cut-set of a cut $\overline C$ in $\overline G$ with cut size at most $c$. 
One can verify that the edge set 
\[E' = \{(v_{u, w}, v_{w, u}, 1): (u, w) \in E'\}\]
is a subset of edges of $E$, and $E'$ is the cut-set of a cut $C$ in $G$. 
The cut size of $C$ in $G$ is the same as the cut size of $\overline C$ in $\overline G$.

Similarly, suppose $E'$ is the cut-set of a cut $C$ in $G$ with cut size at most $c$. Then every edge in $E'$ is an edge of form $(v_{u, w}, v_{w, u}, 1)$ for some $u, w \in \overline V$. 
Let $\overline E'$ be the edge subset defined as 
\[\overline E' = \{(u, v) : (v_{u, w}, v_{w, u}, 1) \in E'\}. \]
$\overline E'$ is the cut-set of a cut $\overline C$ in $\overline G$.
The cut size of $\overline C$ in $\overline G$ is the same as the cut size of $C$ in $G$.
Hence, there is a bijection between all the $c$-cuts in $\overline G$ and all the $c$-cuts in $G$. 

If edge $(u, w)$ is inserted to graph $\overline G$, then we update graph $G$ and the neighbor lists for the vertices of $\overline G$ as follows: Let $u'$ be the last element in the neighbor list of $u$, and $w'$ be the last element in the neighbor list of $w$. 
We apply the following updates to multigraph $G$
\[\insertt(v_{u, w}), \insertt(v_{w, u}), \insertt(v_{u, w}, v_{w, u}, 1), \insertt(v_{u, w}, v_{u, u'}, c+1), \insertt(v_{w, u}, v_{w, w'}, c+1),\]
and add $w$ to the end of the neighbor list of $u$, and add $u$ to the end of the neighbor list of $w$. 

If edge $(u, w)$ is deleted from graph $\overline G$, then we update graph $G$ and the neighbor lists for the vertices of $\overline G$ as follows: Let $u'$ be the previous element of $w$ in the neighbor list of $u$, 
$u''$ be the next element after $w$ in the neighbor list of $u$ if exists, 
$w'$ be the previous element in the neighbor list of $w$, 
and $w''$ be the next element in the neighbor list of $w$ if exists. 
We apply the following updates to multigraph $G$
\[\begin{split}
& \ \ \ \ \delete(v_{u, w}, v_{w, u}, 1),
\delete(v_{u, w}, v_{u, u'}, c+1), 
\delete(v_{u, w}, v_{u, u''}, c+1), 
\delete(v_{w, u}, v_{w, w'}, c+1),\\
& 
\delete(v_{w, u}, v_{w, w''}, c+1),
\insertt(v_{w, w'}, v_{w, w''}, c+1),
\insertt(v_{u, u'}, v_{u, u''}, c+1), 
\delete(v_{u, w}), \delete(v_{w, u}),
\end{split}\]
where the operations related to $u''$ (or $w''$) are omitted if $u''$ (or $w''$) does not exist. 
We also remove $w$ from the neighbor list of $u$, and remove $u$ from the neighbor list of $w$.

\section{Contraction Technique}\label{sec:contraction}

In this section, we summarize the construction  $\contract_K(F)$ for a given unweighted forest $F$ and a set of terminals $K$.

We start by defining the \emph{connecting paths} for a tree $T = (V, E)$ and a set of terminals $K \subset V$. 
The set of connecting paths of $T$ with respect to $K$ is the minimal collection of edge disjoint paths in $T$ that connect vertices in $K$.
Formally, given $T$ and $K$, 
the set of connecting paths, denoted as $\mathsf{Path}_{K}(T)$,  is a set of edge disjoint paths of $T$ satisfying the following conditions:
\begin{enumerate}
    \item The union of all the paths in  $\mathsf{Path}_{K}(T)$ is a connected subtree of $T$. 
    \item For any $v \in K$, $v$ is an end point of some path $\mathsf{Path}_{K}(T)$. 
    \item For any endpoint $v$ of a path in $\mathsf{Path}_{K}(T)$, 
    $v$ is either a vertex in $K$ or $v$ is an endpoint for at least three paths in $\mathsf{Path}_{K}(T)$. 
\end{enumerate}

It was shown in~\cite{henzinger1997fully} that for fixed $T$ and $K$, $\mathsf{Path}_{K}(T)$ is uniquely defined, and has $O(K)$ paths. 

The contracted tree with respect to $T$ and $K$, denoted as $\contract_K(T)$, is constructed as follows:
\begin{enumerate}
    \item The vertices of $\contract_K(T)$ are the endpoints of paths in $\mathsf{Path}_{K}(T)$. 
    \item There is an edge $(u, v)$ in $\contract_K(T)$ iff $u$ and $v$ are the two endpoints of a path in $\mathsf{Path}_{K}(T)$. 
\end{enumerate}
Without loss of generality, we assume if $|K|=1$,  $\contract_K(T)$ contains only the isolated vertex of $K$, and 
if $K = \emptyset$, $\contract_K(F)$ is an empty graph. 

The contracted tree can be extended to the contracted forest as follows: For an unweighted forest $F = (V, E)$ and a set of terminals $K \subset V$, 
$\contract_K(F)$ is the union of $\contract_{K \cap V(T)}(T)$ for each tree $T$ in $F$.

\begin{lemma}
For an unweighted forest $F = (V, E)$ and a set of terminals $K \subset V$,
$\contract_K(F)$ is a forest that contains at most $O(|K|)$ vertices and edges. 
\end{lemma}

\section{Fully Dynamic Algorithm From Offline Update Algorithm}
\label{sec:online_batch_general}

The following lemma, which was implicitly given in~\cite{nanongkai2017dynamic}, shows that for a data structure of a graph, 
if there is an offline update algorithm (i.e., the update algorithm can update the data structure with respect to a given update sequence of the graph) such that the data structure can be maintained by applying the update algorithm for a bounded number of update sequences since initialized, then there is a fully dynamic algorithm to maintain the data structure. 

\begin{lemma}[Section 5, \cite{nanongkai2017dynamic}]\label{lem:fully_framework}
Let $G$ be a dynamic graph, and $\mathfrak{D}$ be a data structure of $G$.
Let $\parametertimesub$, $ \parameterlength$, $\rti$, and $\rtu$ be four parameters. 
If 
there is an initialization algorithm for $\mathfrak{D}$ with running time $\rti$
and 
an offline update algorithm satisfying the following two conditions:
\begin{enumerate}
    \item the data structure $\mathcal{D}$ can be maintained after applying the offline update algorithm with any $\parametertimesub$ different update sequences sequentially; 
    \item for an update sequence $\updateseq$ for $G$ such that $|\updateseq| \leq \parameterlength$, the running time of the update algorithm is at most $\rtu \cdot |\updateseq|$,
\end{enumerate}
then for any $\parametertimes \leq \parametertimesub$ satisfying $\parameterlength \geq 2\cdot 6^\parametertimes$, 
there is a fully dynamic algorithm  with initialization time $O(2^\parametertimes \cdot \rti)$
and worst-case update time $O\left(4^\parametertimes \cdot \left( \rti  /w  + w^{(1/\parametertimes)}\rtu\right)\right)$
to maintain a set of $O(2^\parametertimes)$ instances of the data structure $\mathfrak{D}$
such that after each update, 
the update algorithm specifies one of the maintained data structure instances satisfying the following conditions
\begin{enumerate}
\item The specified data structure  instance is for 
the up-to-date graph.
\item The online-batch update algorithm is executed for at most $\parametertimes$ times on the specified data structure instance with each update batch of size at most $\parameterlength$.
\end{enumerate}
\end{lemma}

\end{document}